\documentclass[11pt,reqno]{amsart}

\date{March 10, 2013}

\usepackage[ansinew]{inputenc}
\usepackage{textcomp}
\usepackage{cite}

\usepackage{amsthm}
\usepackage{amssymb}
\usepackage{amsfonts}

\newtheorem{theorem}{Theorem}[section]
\newtheorem{proposition}[theorem]{Proposition}
\newtheorem{lemma}[theorem]{Lemma}
\newtheorem{corollary}[theorem]{Corollary}

\theoremstyle{definition}

\theoremstyle{remark}
\newtheorem*{remark}{Remark}

\numberwithin{equation}{section}

\newcommand{\al}{\hat a_{k_3}}
\newcommand{\all}{\hat a_{k_3'}}
\newcommand{\cl}{\mathrm{c}}
\newcommand{\co}{\mathrm{co}}

\newcommand{\K}{\mathcal{K}}

\newcommand{\lk}{\left(}

\newcommand{\N}{\mathbb{N}}
\newcommand{\q}{\mathrm{q}}
\newcommand{\R}{\mathbb{R}}
\newcommand{\rk}{\right)}

\newcommand{\Zv}{\mathbf{Z}}

\DeclareMathOperator*{\esssup}{ess\,sup}
\DeclareMathOperator{\re}{Re}

\begin{document}

\title[Polaron in a strong magnetic field]{The ground state energy of a polaron\\ in a strong magnetic field}

\author[R. Frank]{Rupert L. Frank}
\address{Rupert L. Frank, Department of Mathematics, Caltech, Pasadena, CA 91125, USA}
\email{rlfrank@caltech.edu}

\author[L. Geisinger]{Leander Geisinger}
 \address{Leander Geisinger, Department of Mathematics, Princeton University, Princeton, NJ 08544, USA}
 \email{leander@princeton.edu}

\thanks{\copyright\, 2013 by the authors. This paper may be reproduced, in its entirety, for non-commercial purposes.\\
Work partially supported by NSF grants PHY--1068285 (R.L.F.) and PHY-1122309 (L.G.) and DFG grant GE 2369/1-1 (L.G.).}

\begin{abstract}
We show that the ground state of a polaron in a homogeneous magnetic field $B$ and its energy are described by an effective one-dimensional minimization problem in the limit $B\to\infty$. This holds both in the linear Fr\"ohlich and in the non-linear Pekar model and makes rigorous an argument of Kochetov, Leschke and Smondyrev.
\end{abstract}

\maketitle


\section{Introduction and main results}
\label{}

A central theme in mathematical physics is the derivation of effective equations for a given model in a certain asymptotic regime and the quantification of approximation errors. Remarkably, even when the original model is linear, the effective one often turns out to be non-linear. The purpose of our work here is to derive an effective non-linear one-dimensional equation for the ground state of a polaron in a strong magnetic field.

A polaron describes an electron interacting with the quantized optical modes of a polar crystal, and a `large' polaron refers to the case where the spatial extension of this polaron is large compared with the spacing of the underlying lattice. In this paper we consider a large polaron in the presence of a strong homogeneous magnetic field. This case has been extensively studied in the physics literature, typically under the name `magnetopolaron', and we refer to the surveys \cite{GeLo,De} for references and background information. We shall mention some specific results after having introduced our problem precisely.

We consider two models for a polaron in a magnetic field. The first model, the so-called \emph{Fr\"ohlich model}, involves a quantized phonon field. In this model the polaron energy is described by the Hamiltonian
\begin{equation}
\label{eq:ham}
\mathfrak{h} = H_B -\partial_3^2 + \frac{\sqrt \alpha}{2\pi} \int_{\R^3} \lk  \frac{a_k}{|k|} e^{i k \cdot x} + \frac{a_k^*}{|k|} e^{-i k \cdot x} \rk dk + \int_{\R^3} a^*_k a_k \,dk
\end{equation}
acting in $L^2(\R^3) \otimes \mathcal{F}(L^2(\R^3))$, and the ground state energy is given by the bottom of its spectrum
$$
E_B^\q = \inf \left\{ \lk \Psi, \mathfrak{h} \Psi \rk_{L^2(\R^3) \otimes \mathcal{F}} :  \| \Psi \|_{L^2(\R^3) \otimes \mathcal{F}} = 1 , \Psi \in H^1_A(\R^3) \otimes \mbox{dom} (\sqrt \mathcal N) \right\} \,.
$$
(The superscript $\q$ stands for `quantized'.)

Here $\mathcal F = \mathcal{F}(L^2(\R^3))$ denotes the bosonic Fock space over $L^2(\R^3)$ with creation and annihilation operators $a_k^*$ and $a_k$ satisfying $\left[ a_k, a^*_{k'} \right] = \delta(k-k')$ and $\left[a_k,a_{k'}\right]=\left[a_k^*,a_{k'}^*\right]=0$ for all $k, k' \in \R^3$. The number operator $\mathcal{N} = \int_{\R^3} a^*_k a_k \,dk$ describes the energy of the phonon field. The terms $H_B -\partial_3^2$ describe the kinetic energy of the electron, where
$$
H_B = \left( -i\partial_1 + A_1(x) \right)^2 + \left( -i\partial_2 + A_2(x) \right)^2
$$
denotes the Landau Hamiltonian corresponding to a homogeneous magnetic field of strength $B>0$ pointing in the $x_3$-direction. The vector potential $A$ can be chosen in the symmetric gauge
$$
A(x_1,x_2,x_3) = \frac B2 (-x_2, x_1, 0) \, .
$$
The space $H^1_A(\R^3)$ is the corresponding magnetic Sobolev space of order one. Finally, the parameter $\alpha>0$ in \eqref{eq:ham} describes the strength of the interaction between the electron and the phonon field. We note that our normalization of $\alpha$ differs from the usual one, but makes our formulas easier. For details about the definition of $\mathfrak h$ as a self-adjoint operator in $L^2(\R^3) \otimes \mathcal{F}(L^2(\R^3))$ we refer the reader to \cite{Ne}, see also \cite{MiSp}.

The second model that we consider, the \emph{Pekar model}, involves a classical phonon field. The polaron energy in this model is given by the (non-quadratic) functional 
\begin{equation}
\label{eq:pekarfunc}
\mathcal{E}_B[\phi]  =  (\phi, H_B \phi) + (\phi, -\partial_3^2 \phi) - \frac \alpha 2 \iint_{\R^3\times\R^3} \frac{|\phi(x)|^2 \, |\phi(y)|^2}{|x-y|} \,dx \,dy
\end{equation}
and the ground state energy is defined as
$$
E_B^\cl = \inf \left\{ \mathcal{E}_B[\phi] : \|\phi\| = 1, \phi \in H^1_A(\R^3) \right\} \,.
$$
(The superscript $\cl$ stands for `classical'.)

It is interesting, although not necessary for our argument, that a minimizer for $E_B^\cl$ exists. This was recently shown in \cite{GrHaWe}, generalizing an earlier theorem in \cite{Li} for $B=0$.

To understand the connection between the Pekar functional and the Fr\"ohlich Hamiltonian we note that for every $\phi\in H^1_A(\R^3)$ with $\|\phi\|=1$
\begin{align*}
\mathcal{E}_B[\phi] = \inf_a & \left( (\phi, H_B \phi) + (\phi, -\partial_3^2 \phi) + \frac{\sqrt\alpha}{2\pi} \iint_{\R^3\times\R^3} \left( \frac{a(k)}{|k|} e^{ik\cdot x} + \frac{\overline{a(k)}}{|k|} e^{-ik\cdot x} \right) |\phi(x)|^2 \,dx\,dk  \right. \\
& \quad \left. + \int_{\R^3} |a(k)|^2 \,dk \ \|\phi\|^2 \right) \,,
\end{align*}
where the infimum is taken over all \emph{functions} $a$ on $\R^3$. This observation can be combined with an application of coherent states to show that
\begin{equation}
\label{eq:pekar}
E_B^\q \leq E_B^\cl
\end{equation}
for all $B$ and $\alpha$. This argument is due to Pekar \cite{Pek51,Pek63}.

Our main results are large $B$ asymptotics of both $E_B^\q$ and $E_B^\cl$. We shall prove

\begin{theorem}[Fr\"ohlich model]
\label{thm:operator}
For every fixed $\alpha>0$,
$$
E_B^\q = B - \frac {\alpha^2}{48} (\ln B)^2 + O\lk (\ln B)^{3/2} \rk 
\quad\text{as}\quad B\to\infty \, .
$$
\end{theorem}

In the case of a classical field we are able to identify even a third term in the asymptotic expansion.

\begin{theorem}[Pekar model]
\label{thm:functional}
For every fixed $\alpha>0$,
$$
E_B^\cl = B - \frac {\alpha^2}{48} \lk \ln B  \rk^2 + \frac {\alpha^2}{12}  (\ln B)  (\ln \ln  B) + O \lk \ln B \rk
\quad\text{as}\quad B\to\infty \, .
$$
\end{theorem}

\begin{remark}
By scaling it is easy to see that $E_B^\cl=\alpha^2 \tilde E_{\alpha^{-2}B}^\cl$, where $\tilde E_B$ is the same as $E_B$ but with $\alpha=1$ (see Section \ref{sec:funproof}). Thus, the asymptotics can also be written as
$$
E_B^\cl = B - \frac {\alpha^2}{48} \lk \ln \frac{B}{\alpha^2}  \rk^2 + \frac {\alpha^2}{12}  \left(\ln \frac B{\alpha^2} \right) \left( \ln \ln  \frac B{\alpha^2} \right) + O \lk \alpha^2 \ln \frac B{\alpha^2} \rk
\quad\text{as}\quad \frac B{\alpha^2} \to\infty \, .
$$
\end{remark}

The physics literature contains upper bounds on $E_B^\q$ and $E_B^\cl$ of the form $B-C\alpha^2 (\ln B)^2$ with explicit but non-optimal constants $C>0$; see \cite{WhPaRo,LeMa,Sa}. These bounds are based on trial function computations. In \cite{KoLeSm}, Kochetov, Leschke and Smondyrev have argued that the correct constant for $E_B^\q$ should be $-C= -1/48$ (in our units). Our theorem verifies this prediction rigorously.

Let us explain the physical intuition behind this problem and how the constant $-1/48$ arises. As $B\to \infty$ the motion of the electron is so fast with respect to that of the phonons (which we have fixed to be of order one) that it becomes uncorrelated. This was the approximation in Pekar's inequality \eqref{eq:pekar} and thus we may expect $E_B^\q$ and $E_B^\cl$ to have the same leading order behavior. Because there is no correlation, the electron can be treated separately from the field. For energetic reasons the electron will be confined to the lowest Landau level in the plane orthogonal to the magnetic field. This corresponds to a spatial extension of order $B^{-1/2}$ in this plane. The shape of the electron density with respect to the $x_3$-direction parallel to the magnetic field is most easily understood in the Pekar model. For smooth functions $\rho\geq 0$ on $\R^2$ with $\int_{\R^2}\rho \,dx_\bot=1$ we have
$$
B^2 \iint_{\R^2\times\R^2} \frac{\rho(B^{1/2} x_\bot)\,\rho(B^{1/2} y_\bot)}{\sqrt{(x_\bot-y_\bot)^2+(x_3-y_3)^2}} dx_\bot\,dy_\bot \sim (\ln B) \delta(x_3-y_3)
$$
as $B\to\infty$. (Here we wrote $x=(x_\bot,x_3)\in\R^2\times\R$.) This suggests that the energy due to the motion in the $x_3$-direction is given by the one-dimensional \emph{effective Pekar functional}
$$
\int_{\R} |f'|^2 \,dx_3 - \frac{\alpha \ln B}{2} \int_\R |f|^4 \,dx_3 \,.
$$
It turns out that the minimization problem for the latter functional can be solved explicitly and one obtains
$$
\inf\left\{ \int_{\R} |f'|^2 \,dx_3 - \frac{\alpha \ln B}{2} \int_\R |f|^4 \,dx_3:\ \int_\R |f|^2\,dx_3 = 1 \right\} = - \frac{\alpha^2 (\ln B)^2}{48} \,.
$$
This is the desired second term in our Theorems \ref{thm:operator} and \ref{thm:functional}.

Note also that the minimizer of the one-dimensional functional is localized on the scale $(\ln B)^{-1}$. This suggests that the electron density of a polaron in a strong magnetic field has the shape of a prolate ellipsoid with characteristic lengths $B^{-1/2}$ and $(\ln B)^{-1}$.

Remarkably, the one-dimensional polaron was introduced by Gross \cite{Gr} as a toy model for the three-dimensional problem, independently of any connection with magnetic fields. His paper also contains the solution of the one-dimensional minimization problem, although in the mathematical literature it can be traced back at least to \cite{Nag41}. As we have already mentioned, the connection between the one-dimensional polaron and the three-dimensional magnetopolaron is due to Kochetov, Leschke and Smondyrev \cite{KoLeSm}.

The above heuristics emphasize, in particular, that the motion in the direction of the $x_3$-axis differs crucially from the motion in the transverse plane. In the polaron context this observation is attributed to \cite{Ku}. A similar phenomenon occurs in other problems with a strong magnetic field, for example, for the one-electron atom, as treated by Avron, Herbst and Simon \cite{AvHeSi}, or for $N$-electron atoms, as treated by Lieb, Solovej and Yngvason \cite{LieSolYng94} and Baumgartner, Solovej and Yngvason \cite{BaSoYn}.

\bigskip

This paper is organized as follows. For pedagogic reasons we begin with the proof of the energy asymptotics in the Pekar model. In the next section we derive initial estimates on the Coulomb energy. In Section \ref{sec:oned} we explain the relation of these bounds to the above minimization problem in one dimension. We prove Theorem \ref{thm:functional} in Section \ref{sec:funproof}. First we derive an upper bound on $E_B^\cl$ by choosing an appropriate trial function. Then we refine the estimate on the Coulomb energy after projecting to the lowest Landau level and we use this to prove the lower bound on $E_B^\cl$. The claim of Theorem~\ref{thm:functional} follows directly from \eqref{eq:scaling}, Corollary~\ref{cor:upper}, and Theorem~\ref{thm:lower}.

Theorem \ref{thm:operator} about the energy asymptotics in the Fr\"ohlich model are proved in Section \ref{sec:fock} and Section \ref{sec:fockmain}. The upper bound follows immediately from \eqref{eq:pekar}. The proof of the lower bound consists of a reduction to the lowest Landau level, which accomplished in Section \ref{sec:fock}, and the analysis on that level. The main result in the latter analysis is Proposition \ref{pro:fockmain}, which is proved in Section \ref{sec:fockmain}.

\bigskip

\emph{Notation.} The letter $C$ stands for a positive constant whose value may change from line to line. The norm $\|\cdot\|$ denotes the standard norm in $L^2(\R^d)$ where $d=1,2,3$ is clear from the context.


\section{Bounds on the Coulomb energy}
\label{sec:basic}

First we establish some basic estimates and show that the energy functional $\mathcal{E}_B$ from \eqref{eq:pekarfunc} is well-defined on $H^1_A(\R^3)$. Let us introduce the notation
$$
D(\varphi,\phi) = \frac 12 \iint_{\R^3\times\R^3} \frac{\overline{\varphi(x)} \phi(y)}{|x-y|} \,dx\, dy \,.
$$
The Hardy-Littlewood-Sobolev inequality (see, e.g., \cite{LieLos01}),
$$
D ( |\phi|^2 ,|\phi|^2 ) \leq C \| \phi^2 \|^2_{6/5} \,,
$$
together with the H\"older inequality, then the Sobolev inequality, and finally the diamagnetic inequality yields
\begin{equation}
\label{eq:hls}
D ( |\phi|^2 ,|\phi|^2 ) \leq C \|\phi\|^3 \|\phi\|_6 \leq C \|\phi\|^3 \|\nabla|\phi|\| \leq C \| \phi \|^3  \| ( -i \nabla + A ) \phi \|  < \infty
\end{equation}
for $\phi \in H_A^1(\R^3)$. This bound easily implies that $E_B^\cl>-\infty$.

We also record the following bounds for later use. For any $\phi \in H^1_A(\R^3)$ and a.e. $x = (x_\perp, x_3) \in \R^2 \times \R$ we have
\begin{align}
\nonumber
|\phi(x_\perp,x_3)|^2 & = \re \lk \int_{-\infty}^{x_3} \overline{(\partial_t \phi)(x_\perp,t)} \phi(x_\perp,t) dt - \int_{x_3}^\infty \overline{(\partial_t \phi)(x_\perp,t)} \phi(x_\perp,t) dt \rk \\
\label{eq:basic1}
&\leq \lk \int_\R |(\partial_t \phi)(x_\perp,t)|^2 dt \rk^{1/2}  \lk \int_\R | \phi(x_\perp,t)|^2 dt \rk^{1/2}
\end{align}
and, therefore,
\begin{align}
\nonumber
\int_{\R^2} |\phi(x_\perp,x_3)|^2 dx_\perp 
&\leq \lk \int_{\R^2} \int_\R |(\partial_t \phi)(x_\perp,t)|^2 dt\,dx_\bot \rk^{1/2}  \lk \int_{\R^2} \int_\R | \phi(x_\perp,t)|^2 dt\,dx_\bot \rk^{1/2} \\
\label{eq:basic2}
&= \| \partial_3 \phi \| \, \|\phi\|\, .
\end{align}

Following \cite{LieSolYng94} we now prove an estimate on $D ( |\phi|^2 ,|\phi|^2 )$. In Section \ref{sec:oned} we will see that the main term of this bound leads to  the second term of the asymptotics of $E_B^\cl$.
 
\begin{proposition}
\label{pro:lsy}
For $\phi \in H_A^1(\R^3)$ and $B >1$ we have
$$
D ( |\phi|^2 ,|\phi|^2 )= \lk \frac{\ln B}{2} - \ln\ln B \rk \int_\R \left| \int_{\R^2} |\phi(x_\perp,x_3)|^2 dx_\perp \right|^2 dx_3 + R^{(1)}_B(\phi) + R^{(2)}_B(\phi)
$$
with remainder terms
$$
|R^{(1)}_B(\phi)| \leq  \frac{\ln B}2 \| \phi \|^4 + \frac{4}{(\ln B)^{1/2} } \| \phi \|^{5/2} \| \partial_3 \phi \|^{3/2}
$$
and
$$
R^{(2)}_B(\phi) =  \int_\R \iint_{\R^2 \times \R^2} |\phi(x_\perp,x_3)|^2 |\phi(y_\perp,x_3)|^2 K_B(x_\perp-y_\perp) dy_\perp dx_\perp dx_3 \, .
$$
Here $K_B$ is given by
$$
K_B(x_\perp) = \ln \lk 1+ \sqrt{1+(\ln B)|^2 |x_\perp|^2 } \rk - \ln \lk \sqrt B |x_\perp| \rk \, .
$$
\end{proposition}

\begin{proof}
First, we rewrite
\begin{align}
\nonumber
& \frac 12 \iint \frac{|\phi(x)|^2 |\phi(y)|^2}{|x-y|} dy dx - \frac{\ln B}{2} \int |\phi(x)|^2 \int_{\R^2}  |\phi(y_\perp,x_3)|^2 dy_\perp dx \\
\nonumber
& = \frac{\ln B}{2} \int |\phi(x)|^2 \lk \int \frac{|\phi(x+y)|^2}{|y| \ln B } dy - \int_{\R^2} |\phi(x_\perp + y_\perp, x_3)|^2 dy_\perp \rk dx \\ 
\label{eq:remainders}
& =\frac{\ln B}{2} \int |\phi(x)|^2 \lk  r_B^{(1)}(x) + r_B^{(2)}(x) + r_B^{(3)}(x) \rk dx
\end{align}
with 
\begin{align*}
r_B^{(1)}(x) &= \int_{|y_3|\geq 1/\ln B}  \frac{|\phi(x+y)|^2}{|y|\ln B} dy \, ,\\
r_B^{(2)}(x) &=  \int_{|y_3|\leq 1/\ln B}  \frac{|\phi(x+y)|^2 - |\phi(x_\perp + y_\perp,x_3)|^2}{ |y| \ln B} dy \, , \\
r_B^{(3)}(x) &=  \int_{\R^2} |\phi(x_\perp + y_\perp,x_3)|^2 \lk \int_{|y_3|\leq 1/\ln B} \frac{1}{ |y| \ln B} dy_3 - 1 \rk dy_\perp \, .
\end{align*}
We immediately see that for all $x \in \R^3$
\begin{equation}
\label{eq:lsy1}
|r_B^{(1)}(x)| \leq \| \phi \|^2 \, .
\end{equation}

To estimate $r_B^{(2)}(x)$ we use the fact that
\begin{align*}
& \left| |\phi(x+y)|^2 - |\phi(x_\perp + y_\perp,x_3)|^2  \right| \\
& \leq  |\phi(x_\perp + y_\perp,x_3+y_3) -\phi(x_\perp + y_\perp,x_3)|   |\phi(x_\perp + y_\perp,x_3+y_3) +\phi(x_\perp + y_\perp,x_3)|  \, .
\end{align*}
The first factor is bounded by
$$
 \sqrt {|y_3|} \lk \int_\R \left| \partial_3 \phi(x_\perp+y_\perp,t)\right|^2 dt \rk^{1/2} 
$$
and by \eqref{eq:basic1} the second factor is bounded by 
$$
2 \lk \int_\R \left| \partial_3 \phi(x_\perp+y_\perp,t)\right|^2 dt \rk^{1/4} \lk \int_\R \left| \phi(x_\perp+y_\perp, t ) \right|^2 dt \rk^{1/4} \, .
$$
It follows that for a.e. $x \in \R^3$
\begin{align*}
|r_B^{(2)}(x)| \leq  & \, \frac{2}{\ln B} \int_{\R^2}  \lk \int_\R \left| \partial_3 \phi(y_\perp,t)\right|^2 dt \rk^{3/4} \lk \int_\R \left| \phi(y_\perp, t ) \right|^2 dt \rk^{1/4} dy_\perp \\
& \times \int_{|y_3|\leq 1/\ln B} \frac{1}{\sqrt{|y_3|}} dy_3 \\
= & \, \frac{8}{(\ln B)^{3/2}} \int_{\R^2}  \lk \int_\R \left| \partial_3 \phi(y_\perp,t)\right|^2 dt \rk^{3/4} \lk \int_\R \left| \phi(y_\perp, t ) \right|^2 dt \rk^{1/4} dy_\perp \,,
\end{align*}
and applying the H\"older inequality yields
\begin{equation}
\label{eq:lsy3}
|r_B^{(2)}(x)| \leq  8(\ln B)^{-3/2} \| \partial_3 \phi \|^{3/2} \| \phi \|^{1/2} \, .
\end{equation}

Finally, to evaluate $r_B^{(3)}(x)$, we calculate
\begin{align*}
 \int_{|y_3|\leq 1/\ln B} \frac{1}{ |y| \ln B} dy_3 
= \frac 2{\ln B} \lk \ln \lk 1 + \sqrt{(\ln B)^2 |y_\perp|^2 + 1 } \rk - \ln |y_\perp|  - \ln\ln B  \rk  \, .
\end{align*} 
Hence, we get
$$ 
 \int_{|y_3|\leq 1/\ln B} \frac{1}{|y| \ln B } dy_3   -1  = \frac{2}{\ln B} (K_B(y_\perp) - \ln \ln B )
$$
and it follows that
\begin{equation}
\label{eq:lsy2}
r_B^{(3)}(x) = \frac{2}{\ln B} \int_{\R^2} |\phi(y_\perp,x_3)|^2 (K_B(x_\perp - y_\perp) - \ln \ln B) dy_\perp \, .
\end{equation}

If we now put
$$
R^{(1)}_B(\phi) = \frac{\ln B}{2}  \int |\phi(x)|^2 \lk r_B^{(1)}(x)+ r_B^{(2)}(x) \rk dx
$$
and
$$
R^{(2)}_B(\phi) = \frac{\ln B}{2} \int |\phi(x)|^2  r_B^{(3)}(x)  dx  + \ln \ln B  \int_\R \left| \int_{\R^2} |\phi(x_\perp,x_3)|^2 dx_\perp \right|^2 dx_3 \,,
$$
then the claim follows from \eqref{eq:remainders}, \eqref{eq:lsy1}, \eqref{eq:lsy3}, and \eqref{eq:lsy2}.
\end{proof}


\section{The one-dimensional functional}
\label{sec:oned}

In order to motivate the material in this section, let us neglect for a moment the remainder terms in Proposition \ref{pro:lsy} and let us assume that $\phi(x_\perp,x_3) = g(x_\perp) f(x_3)$ with $\|g\|_{L^2(\R^2)}  = \| f\|_{L^2(\R)} = 1$. Then Proposition \ref{pro:lsy} would imply
$$
\mathcal{E}_B[\phi] \sim (g,H_B g) + \int_\R |(\partial_3 f)(x_3)|^2 dx_3 - C_B  \int_\R |f(x_3)|^4 dx_3 
$$
with $C_B = (\ln B)/2 - \ln \ln B$. It is well known that the Landau Hamiltonian satisfies
$$
\inf_{\|g\| = 1} (g,H_B g) = B \, .
$$
Thus, to prove our result under the simplifying assumptions made above it would remain to establish that the infimum  of the one-dimensional functional $\int_\R |f'(t)|^2 dt - C_B \int_\R |f(t)|^4 dt$ is given by $-C_B^2/12$. In fact, this result is implicitly contained in \cite{Nag41}. We formulate this result as follows.

\begin{lemma}
\label{lem:oned}
Let $a,b>0$. Then
$$
\inf \left\{ \int_\R |f'(t)|^2 dt - b \int_\R |f(t)|^4 dt \, : \, \int_\R |f(t)|^2 dt = a \right\} = - \frac{b^2}{12} a^3 \,,
$$
and the infimum is attained at 
$$
f_{a,b}(t) = \frac{a \sqrt b}{2}  \lk \cosh\frac{abt}2 \rk^{-1} \, .
$$
\end{lemma}

\begin{proof}
This follows from the estimate \cite{Nag41}
\begin{equation}
\label{eq:scalinv}
\| g' \|_2^\theta \, \| g\|^{1-\theta}_2 \geq C_q \, \| g \|_q \, , \qquad \theta = \frac 12 - \frac 1q \, , \qquad q >2 \,  ,
\end{equation}
where
$$
C_q = (q\theta)^{-1/q} \lk \frac{2+q\theta}{2q\theta} \rk^{(2+q\theta)/2q} \lk \frac{\sqrt 2 \Gamma(3/2+1/q\theta)}{\Gamma(3/2) \Gamma(1+1/q\theta)} \rk^{-\theta} \, .
$$
For $q = 4$ we have $\theta = 1/4$ and $C_4 = 3^{1/8}$. Given $f$ with $\| f \|_2^2 = a$ we set $f(t) = \sqrt \lambda g(\lambda t)$ for $\lambda > 0$, so that $\| g \|_2^2 = a$ and
$$
 \int_\R |f'(t)|^2 dt - b \int_\R |f(t)|^4 dt = \lambda^2  \int_\R |g'(t)|^2 dt - b \lambda \int_\R |g(t)|^4 dt \geq - \frac {b^2}{4} \frac{\|g\|^8_4}{\| g' \|_2^2} \, ,
$$
where the last estimate follows from minimizing in $\lambda > 0$. From \eqref{eq:scalinv} we learn that $\| g\|_4^8 \|g'\|_2^{-2} \leq C_4^{-8} \|g\|_2^6 = a^3/3$. This yields the first claim. The fact that the infimum is attained at $f_{a,b}$ can be checked by an elementary calculation.
\end{proof}

We conclude that the main term of Proposition \ref{pro:lsy} leads to a sharp lower bound, even if $\phi$ is not given as a product:

\begin{corollary}
\label{cor:secondterm}
For $\phi \in H^1_A(\R^3)$ and $b > 0$
$$
\| \partial_3 \phi \|^2 - b \int_\R \left| \int_{\R^2} |\phi(x_\perp,x_3)|^2 dx_\perp \right|^2 dx_3 \geq - \frac{b^2}{12} \|\phi\|^6 \, .
$$
\end{corollary}

\begin{proof}
For $x_3 \in \R$ let us introduce the function
$$
f(x_3) = \lk \int_{\R^2} |\phi(x_\perp,x_3)|^2 dx_\perp \rk^{1/2} 
$$
such that
$$
\int_\R \left| \int_{\R^2} |\phi(x_\perp,x_3)|^2 dx_\perp \right|^2 dx_3 = \int_\R |f(x_3)|^4 dx_3 \, .
$$
Moreover, by the Schwarz inequality we get $|\partial_3 f(x_3)|^2 \leq  \int_{\R^2} |\partial_3 \phi(x_\perp,x_3)|^2 dx_\perp$ and thus $\int_\R |\partial_3 f(x_3)|^2 dx_3 \leq \| \partial_3 \phi \|^2$. Applying Lemma \ref{lem:oned} yields
$$
\| \partial_3 \phi \|^2 - b \int_\R \left| \int_{\R^2} |\phi(x_\perp,x_3)|^2 dx_\perp \right|^2 dx_3 \geq  \int_\R |\partial_3 f(x_3)|^2 dx_3 - b  \int_\R |f(x_3)|^4 dx_3 \geq - \frac{b^2}{12} \| \phi \|^6.
$$
This finishes the proof.
\end{proof}


\section{Proof of Theorem \ref{thm:functional}}
\label{sec:funproof}

Let us first explain that it suffices to prove Theorem \ref{thm:functional} for $\alpha =1$. To this end we want to make the dependence on $\alpha$ explicit and write $\mathcal E_{B,\alpha}$ and $E_{B,\alpha}^\cl$ for $\mathcal E_{B}$ and $E_{B}^\cl$, respectively.  

Assume that $\phi \in H_A^1(\R^3)$ is normalized. Then $\phi_\alpha(x) = \alpha^{3/2} \phi(\alpha x)$ is also normalized and we have $D( |\phi_\alpha|^2, |\phi_\alpha|^2) = \alpha D(|\phi|^2, |\phi|^2)$ and $(\phi_\alpha, (H_B-\partial_3^2)\phi_\alpha) = \alpha^2 (\phi, (H_{\alpha^{-2}B}-\partial_3^2)\phi)$. Hence, we find $\mathcal E_{B,\alpha}[\phi_\alpha] = \alpha^2 \mathcal E_{\alpha^{-2}B, 1} [\phi]$ and, in particular,
\begin{equation}
\label{eq:scaling}
E^\cl_{B,\alpha} = \alpha^2 E^\cl_{B\alpha^{-2}, 1} \, .
\end{equation}
Thus, for the remainder of this section we assume $\alpha = 1$.

\subsection{The upper bound}
\label{ssec:upper}

The considerations in Section \ref{sec:oned} suggest to derive an upper bound on $E_B$ using the trial  function
$$
\varphi_B(x_\perp,x_3) = \sqrt { \frac B{2\pi} } \exp \lk - \frac B4 |x_\perp|^2 \rk \frac{\sqrt {{|\ln B|}}}{2 \sqrt 2} \lk \cosh \frac{|\ln B|\, x_3}4 \rk^{-1} \, .
$$
Note that this function is of the form $g(x_\perp) f(x_3)$, where $g$ is a ground state of the Landau Hamiltonian $H_B$ and $f = f_{1,|\ln B|/2}$ was introduced in Lemma \ref{lem:oned}.
\begin{corollary}
\label{cor:upper}
There is a constant $C > 0$ such that for $B > 1$ the estimate
$$
E_B^\cl \leq \mathcal E_B[\varphi_B] \leq B - \frac1{48} (\ln B)^2 + \frac 1{12}  (\ln B) ( \ln \ln B) + C \ln B 
$$
holds.
\end{corollary}

\begin{proof}
By elementary calculations we see that $\| \varphi_B \| = 1$ and $(\varphi_B, H_B \varphi_B) = B$. Moreover, the results in Section \ref{sec:oned} show that
\begin{align*}
&\| \partial_3 \varphi_B\|^2 - \lk \frac{\ln B}{2} - \ln \ln B \rk \int_\R \left| \int_{\R^2} |\varphi_B(x_\perp,x_3)|^2 dx_\perp \right|^2 dx_3 \\
&= - \frac{1}{12}  \lk \frac{\ln B}{2} - \ln \ln B \rk ^2 \\
&\leq - \frac{(\ln B)^2}{48} + \frac { (\ln B) (\ln \ln B)}{12} \, .
\end{align*}
Thus, Proposition \ref{pro:lsy} yields
$$
E_B^\cl \leq \mathcal E_B[\varphi_B] \leq B -  \frac{(\ln B)^2}{48} + \frac { (\ln B) (\ln \ln B)}{12}  + |R_B^{(1)}(\varphi_B)| - R_B^{(2)}(\varphi_B) \, .
$$
Since $\|\partial_3 \varphi_B\| \leq C \ln B$, we have $|R^{(1)}_B(\varphi_B)| \leq C  \ln B$. Moreover, we bound $K_B(x_\bot) \geq -\ln_+(\sqrt B |x_\bot|)$ and deduce that $R_B^{(2)}(\varphi_B) \geq -C\ln B$. This proves Corollary \ref{cor:upper}.
\end{proof}

\subsection{The lower bound}
\label{ssec:lower}

In this subsection we supplement the upper bound in Corollary~\ref{cor:upper} with a corresponding lower bound. Theorem \ref{thm:functional} follows directly from these two results.

\begin{theorem}
\label{thm:lower}
There is a constant $C>0$ such that for all $B\geq C$ the estimate
$$
E_B^\cl \geq B - \frac 1{48} (\ln B)^2 + \frac 1{12} (\ln B) (\ln \ln B) - C \ln B 
$$
holds.
\end{theorem}

To derive this estimate we first project in the first two coordinates onto the ground state of the two-dimensional Landau Hamiltonian $H_B$. We recall (see, e.g., \cite{LanLif76}) that the projector onto the lowest Landau level is given by the integral operator $P_0$ in $L^2(\R^2)$ with integral kernel
\begin{equation}
\label{eq:prokernel}
P_0(x_\perp,y_\perp) = \frac{B}{2 \pi} e^{-B|x_\perp-y_\perp|^2/4} e^{iB(x_1 y_2 - x_2 y_1)/2} \, .
\end{equation}
We use the same notation for this operator acting in $L^2(\R^3)$ (and, later, in $L^2(\R^3)\otimes\mathcal F$). Since $P_0$ commutes with $H_B$ and $\partial_3$, we have
\begin{equation}
\label{eq:kinproject}
\|(-i\nabla+A)\phi\|^2 = \|(-i\nabla+A)P_0\phi\|^2 + \|(-i\nabla+A)P_>\phi\|^2 \,, 
\end{equation}
where $P_>=1-P_0$.

We now write $\nabla^\perp=(\partial_1,\partial_2)$ and $A^\perp=(A_1,A_2)$. Since $P_0$ projects onto the lowest Landau level, we have $H_B P_0 = B P_0$ and, thus,
\begin{equation}
\label{eq:b}
\|(-i\nabla^\perp+A^\perp)P_0\phi\|^2 = B \| P_0 \phi \|^2 \, ,
\end{equation}
Moreover, the structure of the spectrum of the Landau Hamiltonian implies that
\begin{equation}
\label{eq:bcomp}
\|(-i\nabla^\perp+A^\perp)P_>\phi\|^2 \geq 3B \| P_> \phi \|^2 \, .
\end{equation}

The following lemma shows how $D ( |\phi|^2 ,|\phi|^2 )$ behaves when we project to the lowest Landau level. For this term there appear off-diagonal terms, however, they can be bounded by the diagonal terms.

\begin{lemma}
\label{lem:project}
There is a constant $C > 0$ such that for all $\phi$ and for all $0 < \tau \leq 1$,
$$
D ( |\phi|^2 ,|\phi|^2 ) \leq  (1+\tau) D ( |P_0\phi|^2 ,|P_0\phi|^2 ) + C \tau^{-3} D ( |P_>\phi|^2 ,|P_>\phi|^2 ) \, .
$$
\end{lemma}

\begin{proof}
First we note that for a.e. $x \in \R^3$ and $\epsilon > 0$
$$
|\phi(x)|^2 = \left| P_0 \phi(x) + P_> \phi (x) \right|^2 \leq (1+\epsilon) |P_0 \phi(x)|^2 +\lk 1+\frac 1\epsilon \rk |P_> \phi(x)|^2 \, .
$$
By definition of $D$ we get for all $\epsilon > 0$
\begin{align}
\nonumber
D ( |\phi|^2 ,|\phi|^2 ) \leq & \, (1+\epsilon)^2 D ( |P_0\phi|^2 ,|P_0\phi|^2 )  + \lk 1+ \frac 1 \epsilon \rk^2 D ( |P_>\phi|^2 ,|P_>\phi|^2 ) \\
\label{eq:dest1}
& + (1+\epsilon) \lk 1 + \frac 1 \epsilon \rk \lk  D ( |P_0\phi|^2 ,|P_>\phi|^2 ) + D ( |P_>\phi|^2 ,|P_0\phi|^2) \rk \, .
\end{align}
To estimate the last term we use the fact that $D$ is positive definite: For all functions $f$ and $g$ in the domain of $D$ we have 
$$
D(f,g) + D(g,f) = D(f,f)+ D(g,g) - D(f-g,f-g) \leq D(f,f) + D(g,g) \, .
$$
We apply this estimate with $f = \sqrt \delta |P_0 \phi|^2$ and $g = \sqrt{ \delta^{-1}} |P_> \phi|^2$ and we obtain, for all $\delta > 0$,
$$
D ( |P_0\phi|^2 ,|P_>\phi|^2 ) + D ( |P_>\phi|^2 ,|P_0\phi|^2) \leq \delta  D ( |P_0\phi|^2 ,|P_0 \phi|^2 ) + \frac 1\delta D ( |P_>\phi|^2 ,|P_>\phi|^2 )  \, .
$$
We can choose for example $\delta = \epsilon^2/(1+\epsilon)$. Then inserting this bound into \eqref{eq:dest1} yields
$$
D ( |\phi|^2 ,|\phi|^2 ) \leq  \, (1+3\epsilon + 2\epsilon^2) D ( |P_0\phi|^2 ,|P_0\phi|^2 )  + ( 1 + \epsilon)^2 (1+2\epsilon) \epsilon^{-3} D ( |P_>\phi|^2 ,|P_>\phi|^2 )
$$
and the claim follows with $\tau = 3\epsilon + 2\epsilon^2$.
\end{proof}

After projecting to the lowest Landau level we want to apply Proposition~\ref{pro:lsy} and Corollary~\ref{cor:secondterm} to estimate $D(|P_0 \phi|^2,|P_0\phi|^2)$. First we need the following bound on the remainder term $R^{(2)}_B(P_0 \phi)$.

\begin{lemma}
\label{lem:lsyprojected}
There is a constant $C>0$ such that for all $\phi \in H^1_A(\R^3)$ and $B>1$
$$
R^{(2)}_B(P_0 \phi) \leq C  \| P_0\phi \|^3  \| \partial_3  P_0\phi \|  \, ,
$$
where $R^{(2)}_B$ was introduced in Proposition \ref{pro:lsy}.
\end{lemma}

\begin{proof}
For $x_\perp \in \R^2$ let $\chi_B(x_\perp)$ denote the characteristic function of the set $\{|x_\perp| \leq 1/\sqrt B \}$.  We decompose $K_B(x_\perp) = K_B^{(1)}(x_\perp) + K_B^{(2)}(x_\perp)$ with $K_B$ defined in Proposition \ref{pro:lsy}. Here we put
$$
K_B^{(1)}(x_\perp) = \ln \lk 1+ \sqrt{1+(\ln B)^2 |x_\perp|^2 } \rk - \ln \lk \sqrt B |x_\perp| \rk \lk 1- \chi_B (x_\perp) \rk
$$
and
$$
K_B^{(2)}(x_\perp) = - \ln \lk \sqrt B |x_\perp| \rk \chi_B (x_\perp) \, .
$$
We note that these kernels satisfy the bounds
$$
K_B^{(1)}(x_\perp)  \leq \ln \left(1+\sqrt{1+B^{-1}(\ln B)^2}\right) \leq C
$$
for all $x_\perp\in\R^2$ and all $B> 1$ and
\begin{equation}
\label{eq:kl2est}
\left\| K_B^{(2)} \right\|_{L^2(\R^2)} = \frac{C}{\sqrt B}
\end{equation}
for all $B>1$ (with a constant $C$ independent of $B$). To estimate the first summand we use \eqref{eq:basic2} and get
$$
\int_\R \iint_{\R^2 \times \R^2} |P_0\phi(x_\perp,x_3)|^2 |P_0\phi(y_\perp,x_3)|^2 K_B^{(1)}(x_\perp-y_\perp) dx_\perp  dy_\perp dx_3  \leq C  \| P_0\phi \|^3  \| \partial_3  P_0\phi \|  \, .
$$
Hence, it remains to estimate $\tilde R^{(2)}_B(P_0 \phi)$ that is defined in the same way as $R^{(2)}_B(P_0 \phi)$ but with $K_B$ replaced by $K^{(2)}_B$. 

Let us fix $x_3 \in \R$ and to simplify notation write $\psi(x_\perp) = P_0\phi(x_\perp,x_3)$ for $x_\perp \in \R^2$. In view of \eqref{eq:kl2est} we can apply the Schwarz inequality to get
\begin{align}
\nonumber
&\iint_{\R^2\times \R^2} |\psi(x_\perp)|^2 |\psi(y_\perp)|^2 K_B^{(2)}(x_\perp-y_\perp) dx_\perp  dy_\perp\\
\nonumber
&= \int_{\R^2} K_B^{(2)}(x_\perp) \int_{\R^2} |\psi(x_\perp+y_\perp)|^2 |\psi(y_\perp)|^2 dy_\perp dx_\perp \\
\label{eq:exest1}
& \leq \frac C{\sqrt B} \| F \|_{L^2(\R^2)} \, ,
\end{align}
where
$$
F(x_\perp) = \int_{\R^2} |\psi(x_\perp+y_\perp)|^2 |\psi(y_\perp)|^2 dy_\perp \, .
$$
Now we estimate
$$
\| F \|_{L^2(\R^2)}^2 \leq  \| F \|_{L^\infty(\R^2)} \int_{\R^2} |F(x_\perp)| dx_\perp \leq \| \psi \|_{L^\infty(\R^2)}^2  \|\psi\|^6_{L^2(\R^2)}
$$
and it remains to estimate $\| \psi \|_{L^\infty(\R^2)}$. The  definition of $\psi$ and the fact that $P_0^2 = P_0$ implies
\begin{align*}
\| \psi \|_{L^\infty(\R^2)} &\leq \esssup_{x_\perp \in \R^2} \int_{\R^2} |P_0(x_\perp,y_\perp)| |P_0 \phi(y_\perp,x_3)| dy_\perp\\
& \leq \esssup_{x_\perp \in \R^2}  \lk \int_{\R^2} |P_0(x_\perp,y_\perp)|^2 dy_\perp \rk^{1/2} \lk \int_{\R^2} |P_0\phi(y_\perp,x_3)|^2 dy_\perp \rk^{1/2}\\
&= \sqrt { \frac B{2\pi} } \lk \int_{\R^2} |P_0\phi(y_\perp,x_3)|^2 dy_\perp \rk^{1/2} \, ,
\end{align*}
where we used the explicit representation of $P_0$, see \eqref{eq:prokernel}, to deduce the last identity.
It follows that
\begin{equation}
\label{eq:exest2}
\|F\|_{L^2(\R^2)} \leq  \sqrt { \frac B{2\pi} } \lk \int |P_0\phi(y_\perp,x_3)|^2 dy_\perp \rk^2  \, .
\end{equation}
Since $x_3 \in \R$ was chosen arbitrarily we can use \eqref{eq:exest1} and  \eqref{eq:exest2} to estimate
\begin{align*}
\tilde R^{(2)}_B(P_0 \phi) &= \int_\R \iint_{\R^2 \times \R^2} |P_0\phi(x_\perp,x_3)|^2 |P_0\phi(y_\perp,y_3)|^2 K_B^{(2)}(x_\perp-y_\perp) dx_\perp  dy_\perp dx_3 \\
& \leq C \int_\R   \lk \int_{\R^2} |P_0\phi(y_\perp,x_3)|^2 dy_\perp \rk^2 dx_3  \, .
\end{align*}
Hence, the claim follows from \eqref{eq:basic2}.
\end{proof}

Now we are in position to prove the lower bound.

\begin{proof}[Proof of Theorem \ref{thm:lower}]
First we project onto the lowest Landau level. We choose $\phi \in H^1_A(\R^3)$ with $\| \phi \| = 1$ and from \eqref{eq:kinproject} and Lemma~\ref{lem:project} we get
\begin{align*}
\mathcal{E}_B[\phi] \geq & \, \|(-i\nabla^\perp+A^\perp)P_0\phi\|^2 + \|\partial_3 P_0\phi\|^2 + \|(-i\nabla +A)P_>\phi\|^2 +\|\partial_3 P_>\phi\|^2 \\
& -  (1+\tau)  D(|P_0\phi|^2,|P_0 \phi|^2)  - C \tau^{-3}  D(|P_>\phi|^2, |P_>\phi|^2) 
\end{align*}
for all $0 < \tau \leq 1$. We insert \eqref{eq:b} and use \eqref{eq:hls} to estimate $ D(|P_>\phi|^2, |P_>\phi|^2) $. After rewriting $B \|P_0 \phi\|^2 = B - B\|P_> \phi\|^2$ we have
\begin{align}
\nonumber
\mathcal{E}_B[\phi] \geq & \, B +  \| \partial_3 P_0 \phi \|^2 -  (1+\tau) D(|P_0\phi|^2,|P_0 \phi|^2)- B \| P_> \phi\|^2  \\
\label{eq:lower1}
& +  \|(-i\nabla^\perp+A^\perp)P_>\phi\|^2 +  \| \partial_3 P_> \phi \|^2 - C \tau^{-3}    \|P_> \phi \|^3 \| (-i\nabla+A) P_> \phi\| \, .
\end{align}

Let us first estimate the terms that involve $P_0 \phi$. We introduce a small parameter $0<\epsilon \leq 1$ and recall the notation $C_B = (\ln B)/2 - \ln \ln B$. From Proposition \ref{pro:lsy} and Lemma \ref{lem:lsyprojected} it follows that 
\begin{align}
\nonumber
& \| \partial_3 P_0 \phi \|^2  -  (1+\tau) D(|P_0\phi|^2,|P_0 \phi|^2)\\
\nonumber
& \geq  (1-\epsilon) \| \partial_3 P_0 \phi \|^2 - (1+\tau) \, C_B \int_\R \left| \int_{\R^2} |P_0\phi(x_\perp,x_3)|^2 dx_\perp \right|^2 dx_3 \\
\label{eq:lower2}
& \quad +\epsilon \| \partial_3 P_0 \phi \|^2 - C (1+\tau) \lk  \ln B + \frac{1}{(\ln B)^{1/2} } \| \partial_3 P_0 \phi \|^{3/2} +  \| \partial_3 P_0 \phi \|  \rk \, .
\end{align}
Here we used the fact that $\|P_0 \phi \| \leq \| \phi \| = 1$ to simplify the error term.
By Corollary~\ref{cor:secondterm} we have
$$
 (1-\epsilon) \| \partial_3 P_0 \phi \|^2 - (1+\tau)C_B \int_\R \left| \int_{\R^2} |P_0\phi(x_\perp,x_3)|^2 dx_\perp \right|^2 dx_3 \geq - \frac{(1+\tau)^2}{1-\epsilon} \frac{C_B^2}{12} \| P_0 \phi \|^6 \, .
$$
If we choose $\epsilon$ and $\tau$ bounded by $1$ and comparable to $(\ln B)^{-1}$ for large $B$, we see that the coefficient on the right-hand side is bounded below by 
$$
-\frac{C_B^2}{12} - C \frac{C_B^2}{\ln B} \geq - \frac{(\ln B)^2}{48} + \frac{(\ln B)(\ln \ln B)}{12} - C \ln B \, .
$$
We claim that with  this choice of $\epsilon$ and $\tau$, all terms in the last line of \eqref{eq:lower2} are bounded below by $- C  \ln B$. Indeed, we can minimize in $\| \partial_3 P_0 \phi \|$. In particular, we find
$$
\frac{\epsilon}{2} \| \partial_3 P_0 \phi \|^2 - C \frac{1+\tau}{(\ln B)^{1/2}} \| \partial_3 P_0 \phi \|^{3/2} \geq - C \frac{(1+\tau)^4}{\epsilon^3 (\ln B)^2} \geq - C \ln B
$$
and
$$
\frac{\epsilon}{2} \| \partial_3 P_0 \phi \|^2 - C (1+\tau) \| \partial_3 P_0 \phi \| \geq - C \frac{(1+\tau)^2}{\epsilon} \geq - C \ln B\, .
$$
Combining these estimate with \eqref{eq:lower2} we arrive at 
\begin{equation}
\label{eq:p0terms}
\| \partial_3 P_0 \phi \|^2  - (1+\tau) D(|P_0\phi|^2,|P_0 \phi|^2)\geq -  \frac{(\ln B)^2}{48}  + \frac{(\ln B) (\ln \ln B)}{12}- C \ln B \, .
\end{equation}

It remains to show that all terms of \eqref{eq:lower1} that involve $P_> \phi$ are bounded below by $-C \ln B$. We introduce another small parameter $0 < \rho \leq 1$ and use the bound \eqref{eq:bcomp} to estimate 
\begin{align*}
& \| (-i\nabla^\perp+A^\perp) P_> \phi \|^2 - B \| P_> \phi\|^2  - C \tau^{-3} \| P_> \phi \|^3 \|(-i\nabla^\perp+A^\perp) P_> \phi \| \\
& = (1-\rho)  \| (-i\nabla^\perp+A^\perp) P_> \phi \|^2 - B \| P_> \phi\|^2 + \rho  \| (-i\nabla^\perp+A^\perp) P_> \phi \|^2 \\
& \quad\quad -  C \tau^{-3} \| P_> \phi \|^3 \| (-i\nabla^\perp+A^\perp) P_> \phi \| \\
& \geq \left((1-\rho)3B - B\right) \| P_> \phi \|^2 + \rho  \| (-i\nabla^\perp+A^\perp) P_> \phi \|^2 -  C \tau^{-3} \| P_> \phi \|^3 \| (-i\nabla^\perp+A^\perp) P_> \phi \| \\
& \geq B \left(2-3\rho\right) \| P_> \phi \|^2 - C \tau^{-6} \rho^{-1} \| P_> \phi \|^6 \, ,
\end{align*}
where the last estimate follows from minimizing in $\| (-i\nabla^\perp+A^\perp) P_> \phi \|$. Since $\| P_> \phi \| \leq 1$ and since $\tau$ is comparable to $(\ln B)^{-1}$ we can choose $\rho$ comparable to $\| P_> \phi \|^2/(\tau^3 \sqrt B)$ and get
\begin{align*}
& \| (-i\nabla^\perp+A^\perp) P_> \phi \|^2 - B \| P_> \phi\|^2  - C \tau^{-3} \| P_> \phi \|^3 \| (-i\nabla^\perp+A^\perp) P_> \phi \| \\
& \geq 2B\| P_> \phi \|^2- C \sqrt B \| P_> \phi \|^4 \tau^{-3} \, .
\end{align*}
To estimate the remaining terms of \eqref{eq:lower1} we note that
$$
\| \partial_3 P_> \phi \|^2 - C \tau^{-3}\|P_> \phi\|^3 \| \partial_3 P_> \phi\| \geq -C \tau^{-6} \|P_> \phi \|^6 \, .
$$
Thus all terms of \eqref{eq:lower1} that involve $P_> \phi$ are bounded below by
$$
2B\| P_> \phi \|^2- C \lk \frac{\sqrt B \| P_> \phi \|^4}{ \tau^3}  + \frac{ \|P_> \phi \|^6}{ \tau^6} \rk \geq \| P_> \phi \|^2 \lk 2B -C \lk \sqrt B (\ln B)^3 + (\ln B)^6 \rk \rk \, ,
$$
since $\| P_> \phi \| < 1$ and $\tau$ is comparable to $(\ln B)^{-1}$. For large $B$ the right-hand side is positive. This finishes the proof.
\end{proof}

\begin{remark}
The above proof also gives bounds on almost minimizers. More precisely, for any $M>0$ there is a constant $C_M>0$ such that for all $B\geq e$ and all $\phi \in H^1_A(\R^3)$ with
$$
\mathcal E_B[\phi] \leq B - \frac{1}{48} (\ln B)^2 + \frac1{12}(\ln B)(\ln\ln B) + M \ln B
$$
one has $\|\partial_3 P_0\phi\|^2 \leq C_M(\ln B)^2$ and
$$
\|(-i\nabla^\perp+A^\perp)P_>\phi\|^2 + B \|P_>\phi\|^2 + \|\partial_3 P_>\phi\|^2 \leq C_M \ln B \,.
$$
Indeed, under the almost minimizing assumption all the error terms in the proof of the lower bound are bounded by a constant times $\ln B$. This easily leads to the stated bounds.
\end{remark}


\section{The ground state energy of the operator $\mathfrak{h}$}
\label{sec:fock}

In this section we outline the proof of Theorem \ref{thm:operator} and reduce it to the proof of Proposition~\ref{pro:fockmain}, which is the topic of the following section.

The upper bound for $E_B^\q$ stated in Theorem \ref{thm:operator} follows from the simple fact that $E_B^\q$ is always bounded from above by $E_B^\cl$; see \cite{Pek51,Pek63} and the discussion in the introduction. Hence, from Theorem \ref{thm:functional}, we obtain that there is a constant $C > 0$ such that for all $B \geq C$  
\begin{equation}
\label{eq:fockupper}
E_B^\q \leq B - \frac{\alpha^2}{48} (\ln B)^2 + \frac{\alpha^2}{12}(\ln B)(\ln\ln B) + C \ln B \, .
\end{equation}

We proceed to the proof of the lower bound for $E_B^\q$. The first step in the proof is to introduce a cut-off in phonon space. For $k \in \R^3$ we write $k = (k_\perp,k_3) \in \R^2 \times \R$ and for   a parameter $\K > 8\alpha/\pi$ we set $\Gamma_\K = \{ k \in \R^3 \, : \, \max(|k_\perp|,|k_3|)\leq \K \}$. Then we introduce the operator
\begin{align}
\mathfrak{h}_{\K}^\co = & \lk 1-\frac{8 \alpha}{\pi \K} \rk \lk H_B -\partial_3^2 \rk 
+  \frac{\sqrt \alpha}{2\pi} \int_{\Gamma_\K}  \lk  \frac{a_k}{|k|}  e^{ik \cdot x} + \frac{a_k^*}{|k|}  e^{-ik \cdot x}  \rk dk \notag \\
\label{eq:ophk}
& + \frac 12 \int_{\R^3} a_k^* a_k \,dk + \frac 12 \int_{\Gamma_\K} a_k^* a_k \,dk   \, .
\end{align}
(The superscrip `co' stands for cut-off.) We shall prove

\begin{lemma}
\label{lem:firstcutoff}
For any $\K> 8\alpha/\pi$ we have $\mathfrak{h} \geq \mathfrak{h}_{\K}^\co - 1/4$.
\end{lemma}

\begin{proof}
We follow the strategy developed in \cite{LieYam58}. For $j = 1,2,3$ we write 
$$
Z_j = \frac{\sqrt \alpha}{2\pi} \int_{\Gamma_\K^c} \frac{k_j}{|k|^3}  a_k e^{ik \cdot x} dk 
$$
with $\Gamma_\K^c = \R^3 \setminus \Gamma_\K$, 
so that
\begin{equation}
\label{eq:firstcutoff1}
\frac{\sqrt \alpha}{2\pi}  \int_{\Gamma_\K^c} \lk  \frac{a_k}{|k|} e^{ik \cdot  x} + \frac{a_k^*}{|k|} e^{-ik \cdot  x}  \rk dk =  \sum_{i=1}^3 \left[ -i\partial_j + A_j , Z_j - Z_j^* \right] \,.
\end{equation}
The expectation $\left< \cdot \right>$ in any (normalized) state satisfies 
\begin{align*}
- \sum_{j=1}^3 \left< \left[ -i \partial_j-A_j, Z_j-Z_j^* \right] \right>  &\leq  2 \left< H_B -\partial_3^2 \right>^{1/2} \left<  - (\Zv-\Zv^*)^2 \right>^{1/2} \\
&\leq   2 \left< H_B -\partial_3^2 \right>^{1/2} \left< 2(\Zv^* \Zv + \Zv \Zv^*) \right>^{1/2} \, ,
\end{align*}
where $\Zv$ denotes the vector $(Z_1,Z_2,Z_3)$.
It follows that for all $\tau > 0$
$$
-   \sum_{i=1}^3 \left[ -i\partial_j + A_j , Z_j - Z_j^* \right]  \leq \tau \lk H_B  -\partial_3^2 \rk +  \frac{2}{\tau} \lk \Zv^* \Zv + \Zv \Zv^* \rk
$$
and we claim that
\begin{equation}
\label{eq:firstcutoffz}
 \Zv^* \Zv + \Zv \Zv^*  \leq \frac{2 \alpha}{\pi \K}\lk \int_{\Gamma_\K^c}  a_k^* a_k  dk + \frac 1 2 \rk \, .
\end{equation}
Combining these estimates with \eqref{eq:firstcutoff1} and choosing $\tau = 8\alpha/\pi \K$ we obtain
$$
\frac{\sqrt \alpha}{2\pi} \int_{\Gamma_\K^c} \lk  \frac{a_k}{|k|}  e^{ik \cdot x} + \frac{a_k^*}{|k|} e^{-ik \cdot  x} \rk dk \geq - \frac{8\alpha}{\pi \K} \lk H_B - \partial_3^2 \rk -\frac 12 \int_{\Gamma_\K^c} a_k^* a_k  dk - \frac 14 \,,
$$
which is the claimed lower bound.

Hence, it remains prove \eqref{eq:firstcutoffz}. By definition, 
$$
\left< \Zv^* \Zv \right>  = \frac{\alpha}{4 \pi^2} \int_{\Gamma_\K^c} \int_{\Gamma_\K^c} \frac{k \cdot k'}{|k|^3 |k'|^3}  e^{i(k'-k) \cdot x} \left< a_k^* a_{k'} \right> dk dk' \, .
$$
We estimate $\langle a_k^* a_{k'} \rangle \leq  \langle  a_k^*  a_k \rangle^{1/2} \langle a_{k'}^* a_{k'} \rangle^{1/2}$ and apply the Schwarz inequality to get
\begin{align*}
\left< \Zv^* \Zv \right> & \leq \frac{\alpha}{4\pi^2} \lk \int_{\Gamma_\K^c} \frac1{|k|^2}  \left<  a_k^*  a_k \right>^{1/2} dk \rk^2 \leq \frac{\alpha}{4\pi^2}  \int_{\Gamma_\K^c} \frac 1{|k|^4} dk \int_{\Gamma_\K^c} \left< a_k^* a_k \right> dk \\
&\leq   \frac{\alpha}{\pi \K} \int_{\Gamma_\K^c} \left< a_k^* a_k \right> dk \, .
\end{align*}
To estimate $\Zv \Zv^*$ we note that $a_k a_{k'}^* =  a_{k'}^*  a_k + \delta(k-k')$ and we can argue in the same way as above. This establishes \eqref{eq:firstcutoffz} and completes the proof.
\end{proof}

Next, we prove a lower bound on $\mathfrak{h}$ which shows already the correct order of the second term of $E_B^\q$. Later, we will use this to estimate the contribution of states that are not in the lowest Landau level. Recall that $P_0$ denotes the projection onto the lowest Landau level, see \eqref{eq:prokernel}, and that $P_>=1-P_0$. We also write $P_0$ for the operator $P_0 \otimes 1\otimes 1$ in $L^2(\R^2) \otimes L^2(\R) \otimes \mathcal F$. 

\begin{lemma}
\label{lem:roughlower}
There is a constant $C > 0$ such that for all $B\geq C$
$$
\mathfrak{h} \geq B P_0 + 3B P_> + \frac 12 \int_{\R^3} a_k^* a_k \,dk - C(\ln B)^2 \, .
$$
\end{lemma}

\begin{proof}
To prove this estimate we need to treat phonon modes in the $k_3$-direction differently from modes in the $k_1$- and $k_2$-directions. 

First, we bound the contribution to $\mathfrak{h}_\K^\co$ that comes from $\{|k| \in \Gamma_\K, |k_3| < \K_3  \}$, where we choose $\K = B / (\ln B)^2$ and $\K_3 = 16 \alpha |\ln B|/\pi$. Note that for all $k \in \R^3$ we have
$$
\lk \frac{a^*_k}{ \sqrt 2} + \frac{\sqrt{2\alpha}}{2\pi |k|} e^{ik\cdot x} \rk \lk \frac{a_k}{ \sqrt 2} + \frac{\sqrt{2\alpha}}{2\pi |k|} e^{-ik\cdot x} \rk \geq 0 \,.
$$
This implies that, for $B$ large enough,
\begin{align*}
& \int_{|k_3| \leq \K_3} \int_{|k_\perp|^2 \leq \K^2} \lk \frac 12 a_k^* a_k + \frac{\sqrt \alpha}{2\pi} \frac {a_k}{|k|} e^{i k \cdot x} + \frac{\sqrt \alpha}{2\pi} \frac{a_k^*}{|k|} e^{-ik \cdot  x}  \rk dk_\perp dk_3 \\
& \geq - \frac{\alpha}{2\pi^2} \int_{|k_3| \leq \K_3} \int_{|k_\perp|^2 \leq \K^2} \frac 1 {|k|^2} dk_\perp dk_3 \\
& = - \frac{\alpha}{\pi} \int_0^{\K_3}  \ln \lk \frac{\K^2 + k_3^2}{k_3^2} \rk  dk_3  \\
& \geq - \frac{2\alpha}{\pi} \K_3  |\ln \K| \, .
\end{align*}
We combine this estimate with Lemma \ref{lem:firstcutoff} and we see that for $B$ large enough 
\begin{align}
\mathfrak{h} \geq  & \lk 1-\frac{8\alpha (\ln B)^2}{\pi B} \rk \lk H_B -\partial_3^2 \rk + \int \lk \frac{\sqrt \alpha}{2\pi}  \frac{a_k}{|k|}  e^{ik \cdot x} + \frac{\sqrt \alpha}{2\pi}  \frac{a_k^*}{|k|}  e^{-ik \cdot x}  \rk dk  \notag \\
\label{eq:firstrough}
& + \frac 12 \int_{\R^3} a_k^* a_k \,dk + \frac 12 \int a_k^* a_k \,dk
 - C (\ln B)^2  \, .
\end{align}
Here and in the remainder of this proof all integrals without specified domain of integration are over $\{ |k_\perp| \leq \K, \K_3 \leq |k_3| \leq \K \}$.

Now we proceed similarly as in Lemma \ref{lem:firstcutoff}. Here we set
$$
Z = \frac{\sqrt \alpha}{2\pi} \int \frac{a_k}{|k| k_3} e^{ik\cdot x} dk
$$ 
and obtain
$$
- \frac{\sqrt \alpha}{2\pi} \int \lk \frac{a_k}{|k|} e^{ik\cdot x} + \frac{a_k^*}{|k|} e^{-ik\cdot x} \rk \leq \frac 12 \lk -\partial_3^2 \rk + 4 \lk Z^* Z + Z Z^* \rk \, .
$$
Applying the Schwarz inequality in the same way as above yields 
$$
4(Z^* Z + Z Z^*) \leq \frac{2\alpha}{\pi^2} \int \frac1{k_3^2} \frac1{|k|^2} dk  \lk \int a_k^* a_k dk+ \frac 12 \rk  \leq \frac 1 2  \lk \int a_k^* a_k dk+ \frac 12 \rk \, ,
$$
where we used the estimate
\begin{align*}
&\int_{\K_3 \leq |k_3| \leq \K} \int_{|k_\perp|^2 \leq \K^2} \frac1{k_3^2} \frac1{|k|^2} dk_\perp dk_3\\
 &= \pi \int_{\K_3}^\K \frac 1{k_3^2} \ln \lk \frac{k_3^2 + \K^2}{k_3^2} \rk dk_3 \\
& = \pi \lk \frac \pi {2\K} + \frac{\ln (\K^2 + \K_3^2)}{\K_3} - \frac{\ln (2\K^2)}{\K} - \frac{2 \arctan(\K_3/\K)}{\K} \rk \\
&\leq  \frac {\pi^2}{4\alpha} \, ,
\end{align*}
valid for $B$ large enough. We put these estimates together and from \eqref{eq:firstrough} we obtain, for $B \geq C$,
$$
\mathfrak{h} \geq \lk 1-\frac{8\alpha (\ln B)^2}{\pi B} \rk  H_B + \lk \frac 12 - \frac{8\alpha (\ln B)^2}{\pi B} \rk \lk -\partial_3^2 \rk + \frac 12 \int_{\R^3} a_k^* a_k \,dk - C (\ln B)^2 \, . 
$$

It remains to note that, for $B$ large enough, we have
$$
\lk \frac 12 - \frac{8\alpha (\ln B)^2}{\pi B} \rk \lk -\partial_3^2 \rk  \geq 0
$$
and
$$
\lk 1-\frac{8\alpha (\ln B)^2}{\pi B} \rk  H_B = \lk 1-\frac{8\alpha (\ln B)^2}{\pi B} \rk  \lk H_B P_0 + H_B P_> \rk \geq B P_0 + 3B P_>  - C (\ln B)^2 \, .
$$
This completes the proof.
\end{proof}

Now we combine the previous lemma with the upper bound \eqref{eq:fockupper} on $E_B^q$. Note that this bound ensures that for every $M>-\alpha^2/48$ there are states $\Psi  \in H^1_A(\R^3) \otimes \textnormal{dom} (\sqrt \mathcal N)$ that satisfy  
\begin{equation}
\label{eq:roughupper}
(\Psi, \mathfrak{h} \Psi )_{L^2(\R^3) \otimes \mathcal F} \leq B + M(\ln B)^2 \, , \quad \|\Psi\|_{L^2(\R^3) \otimes \mathcal F} = 1 \,.
\end{equation}

\begin{corollary}
\label{cor:simplebounds}
For every $M\in\R$ there is a constant $C_M>0$ such that for every $B\geq C_M$ and every $\Psi  \in H^1_A(\R^3) \otimes \textnormal{dom} (\sqrt \mathcal N)$ satisfying \eqref{eq:roughupper} one has
$$
 \| P_> \Psi \|_{L^2(\R^3) \otimes \mathcal F}^2  \leq C_M (\ln B)^2 B^{-1}
$$
and
$$
(\Psi , \mathcal N \Psi)_{L^2(\R^3) \otimes \mathcal F} \leq C_M (\ln B)^2 \, .
$$
\end{corollary}

\begin{proof}
If we combine the lower bound derived in Lemma \ref{lem:roughlower} with the upper bound \eqref{eq:roughupper} we obtain
$$
B \|P_0 \Psi \|^2 + 3B \|P_> \Psi\|^2 + \frac 12 (\Psi, \mathcal N \Psi) \leq B + (M+C) (\ln B)^2
$$
for $B\geq C$. Thus, the claim follows from the identity $\| P_0 \Psi \|^2 = 1 - \|P_> \Psi \|^2$ and from the fact that $\mathcal N$ is non-negative.
\end{proof}

Given the bounds of Corollary \ref{cor:simplebounds} we can reduce the problem to the lowest Landau level. 
The reduction lemma reads as follows.

\begin{lemma}
\label{lem:reduction}
There is a constant $C > 0$ such that the following holds. For every $M\in\R$ there is a $C_M>0$ such that for every $B\geq C_M$ and every $\Psi  \in H^1_A(\R^3) \otimes \textnormal{dom} (\sqrt \mathcal N)$ satisfying \eqref{eq:roughupper} one has for every $C \leq \K \leq C^{-1} B$ 
$$
\lk \Psi, \mathfrak{h} \Psi \rk_{L^2(\R^3) \otimes \mathcal{F}} \geq \lk P_0 \Psi, \mathfrak{h}_\K^\co P_0 \Psi \rk_{L^2(\R^3) \otimes \mathcal{F}}  + B \| P_> \Psi \|_{L^2(\R^3) \otimes \mathcal F}- C_M (\ln B)^2 \sqrt {\K B^{-1}} - \frac 14 \, .
$$
\end{lemma}

\begin{proof}
We note that the operators $H_B$, $-\partial_3^2$, and $\mathcal N$ all commute with $P_0$. Hence, we can apply Lemma \ref{lem:firstcutoff} to estimate
\begin{align}
\nonumber
\mathfrak{h} \geq \, & P_0 \mathfrak h_\K^\co P_0 + P_> \mathfrak h_\K^\co P_> + P_0 \frac{\sqrt \alpha}{2\pi} \int_{\Gamma_\K} \lk \frac{a_k}{|k|} e^{ik \cdot x} + \frac{a_k^*}{|k|} e^{-ik \cdot x} \rk dk P_>\\
\label{eq:hdecomp}
&  +  P_> \frac{\sqrt \alpha}{2\pi} \int_{\Gamma_\K} \lk \frac{a_k}{|k|} e^{ik \cdot x} + \frac{a_k^*}{|k|} e^{-ik \cdot x} \rk dk P_0  - \frac 14 \, .
\end{align}
We estimate the terms on the right side individually, first the diagonal term $P_> \mathfrak h_\K^\co P_>$. For a lower bound we complete the square in the interaction term. Similarly as in the proof of Lemma \ref{lem:roughlower} we have, for $k \in \R^3$,
$$
a^*_k a_k +  \frac{\sqrt \alpha}{2\pi} \frac{a_k}{|k|} e^{ik \cdot x} + \frac{\sqrt \alpha}{2\pi} \frac{a_k^*}{|k|} e^{-ik \cdot x} \geq - \frac {\alpha}{4 \pi ^2|k|^2} 
$$
and we find
$$
P_> \mathfrak h_\K^\co P_> \geq \lk 1- \frac{8\alpha}{\pi \K} \rk H_B P_> - \int_{\Gamma_\K} \frac{\alpha}{4 \pi^2 |k|^2} dk P_> \geq \left( 3B - C B \K^{-1} - C \K \right) P_> \, ,
$$
where we used that
$$
\int_{|k_3| \leq \K} \int_{|k_\perp| \leq \K} \frac 1{|k|^2} dk_\perp dk_3 =  \pi \K \lk2 \ln(2) +  \pi  \rk \, .
$$
Thus the bounds on $\K$ imply
\begin{equation}
\label{eq:largeterms}
 (\Psi,P_> \mathfrak{h}_\K^\co P_> \Psi )_{L^2(\R^3) \otimes \mathcal F} \geq B \| P_> \Psi \|^2_{L^2(\R^3) \otimes \mathcal F}
\end{equation}
for $B$ large enough.

We proceed to estimating the off-diagonal terms. For fixed $x \in \R^3$ let us define the function
$$
f_{\K,x}(k) = \frac{\sqrt \alpha}{2\pi |k|} e^{i k \cdot x} \chi_\K (k) \, , \quad k \in \R^3 \, ,
$$
where $\chi_\K$ denotes the characteristic function of $\Gamma_\K$. Note that $f_{\K,x}$ is in $L^2(\R^3)$ with $\| f_{\K,x} \| \leq C \sqrt \K$, independent of $x$. Hence, we can rewrite the operator
\begin{equation}
\label{eq:atogether}
\frac{\sqrt \alpha}{2\pi} \int_{\Gamma_\K} \frac{a_k}{|k|} e^{ik\cdot x} dk = a \lk f_{\K,x} \rk \, .
\end{equation}
Let us recall the bound
\begin{equation}
\label{eq:afockest}
\| a^*(f_{\K,x}) \Phi \|_{\mathcal F}  \leq \| f_{\K,x} \| \| \sqrt { \mathcal N +1} \Phi \|_{\mathcal F} \leq C \sqrt{\mathcal K} \| \sqrt { \mathcal N+1}  \Phi \|_{\mathcal F}  \, ,
\end{equation}
valid for all $\Phi \in \mbox{dom}(\sqrt \mathcal N)$. Using notation \eqref{eq:atogether} we can write
\begin{align*}
\lk \Psi, P_0 \frac{\sqrt \alpha}{2\pi} \int_{\Gamma_\K} \frac{a_k}{|k|} e^{ik\cdot x}  dk \, P_> \Psi \rk_{L^2(\R^3) \otimes \mathcal F} 
& = \int_{\R^3} \lk (P_0 \Psi)(x), a \lk f_{\K,x} \rk (P_> \Psi)(x) \rk_{\mathcal F} dx \\
& = \int_{\R^3} \lk a^* \lk f_{\K,x} \rk (P_0 \Psi)(x), (P_> \Psi)(x) \rk_{\mathcal F} dx
\end{align*}
for any $\Psi \in L^2(\R^3) \otimes \mathcal F$. Now the bound \eqref{eq:afockest} allows us to estimate
\begin{align*}
& \left| \lk \Psi, P_0 \frac{\sqrt \alpha}{2\pi} \int_{\Gamma_\K} \frac{a_k}{|k|} e^{ik\cdot x}  dk \, P_> \Psi \rk_{L^2(\R^3) \otimes \mathcal F} \right|
\leq  \int_{\R^3} \|a^*\lk f_{\K,x} \rk (P_0 \Psi)(x) \|_{\mathcal F} \|(P_> \Psi)(x) \|_{\mathcal{F}} \, dx \\
& \qquad \leq  \| P_> \Psi \|_{L^2(\R^3) \otimes \mathcal F} \lk \int_{\R^3}  \|a^*\lk f_{\K,x} \rk (P_0 \Psi)(x) \|_{\mathcal F}^2 \,dx \rk^{1/2} \\
& \qquad \leq C_M \sqrt \K  \| P_> \Psi \|_{L^2(\R^3) \otimes \mathcal F} \| \sqrt{\mathcal N + 1} P_0 \Psi \|_{L^2(\R^3) \otimes \mathcal F} \, .
\end{align*}
We combine this estimate with Corollary \ref{cor:simplebounds} and obtain that any state $\Psi$ satisfying \eqref{eq:roughupper} also satisfies
\begin{equation}
\label{eq:offdiagonal}
\left| \lk \Psi, P_0 \frac{\sqrt \alpha}{2\pi} \int_{\Gamma_\K}  \frac{a_k}{|k|} e^{ik\cdot x}  dk P_> \Psi \rk_{L^2(\R^3) \otimes \mathcal F} \right| \leq C_M (\ln B)^2 \sqrt{\K B^{-1} }
\end{equation}
for $B$ large enough. Similarly, we can estimate the remaining three interaction terms. Thus, \eqref{eq:hdecomp}, \eqref{eq:largeterms} and \eqref{eq:offdiagonal} yield the claimed lower bound.
\end{proof}

In view of Lemma~\ref{lem:reduction} we can work in the lowest Landau level and we have to find a lower bound on the operator $P_0 \mathfrak{h}_\K^\co P_0$. This is accomplished in the in following proposition, which plays a similar role as Proposition \ref{pro:lsy} in the analysis of the functional $\mathcal{E}_B$ and which lies at the heart of the proof of Theorem \ref{thm:operator}.

The definition of $\mathfrak h_\K^\co$, see \eqref{eq:ophk}, implies
\begin{equation}
\label{eq:simplelowerop}
\mathfrak{h}_{\K}^\co \geq  \kappa \lk H_B -\partial_3^2 \rk 
+  \frac{\sqrt \alpha}{2\pi} \int_{\Gamma_\K}  \lk  \frac{a_k}{|k|}  e^{ik \cdot x} + \frac{a_k^*}{|k|}  e^{-ik \cdot x}  \rk dk +\int_{\Gamma_\K} a_k^* a_k \,dk   
\end{equation} 
with $\kappa=  1- 8 \alpha/\pi \K$. Here we used the fact that the operator $a^*_k a_k$ is non-negative for all $k \in \R^3$.
We note that the following proposition is also valid for the operator on the right-hand side of \eqref{eq:simplelowerop} with an arbitrary choice of $\kappa$, not necessarily the one made above.

\begin{proposition}
\label{pro:fockmain}
There is a constant $C>0$ such that for all $B$, $\kappa$, and $\K$ satisfying $B\geq C$, $C(\ln B)^{-1/2}\leq \kappa\leq C^{-1} \ln B$, and $\K \geq \sqrt B$ one has
\begin{align*}
P_0 \mathfrak{h_{\K}^\co} P_0  \geq \, &  \left( \kappa B -\frac{\alpha^2 (\ln B)^2}{48\kappa} - C \kappa^{-1/2} (\ln B)^{3/2} \right) P_0 \,.
\end{align*}
\end{proposition}

We defer the proof of this result to the following Section \ref{sec:fockmain}. Here we show how to deduce Theorem \ref{thm:operator} from Proposition \ref{pro:fockmain} together with the previous results in this section.

\begin{proof}[Proof of Theorem \ref{thm:operator}]
We have already discussed the proof of the upper bound at the beginning of this section and we now focus on the lower bound. According to the upper bound there are $\Psi\in H^1_A(\R^3) \otimes \textnormal{dom} (\sqrt \mathcal N)$ satisfying \eqref{eq:roughupper} with any fixed $M>-\alpha^2/48$, and it suffices to prove a lower bound on $(\Psi,\mathfrak{h}\Psi)$ for such $\Psi$.

It follows from Lemma \ref{lem:reduction} and Proposition \ref{pro:fockmain} that for such $\Psi$
$$
(\Psi,\mathfrak{h}\Psi) \geq \kappa B  - \lk \frac{\alpha^2 (\ln B)^2}{48\kappa} + C \kappa^{-1/2} (\ln B)^{3/2} \right) \|P_0 \Psi\|^2 - C (\ln B)^2 \sqrt{ \K B^{-1} } - C
$$
with $\kappa= 1-8\alpha/(\pi \K)$, where we choose $\K = B (\ln B)^{-4/3}$. We bound $\|P_0\Psi\|\leq 1$ and obtain
$$
(\Psi,\mathfrak{h}\Psi) \geq B  - \frac{\alpha^2 (\ln B)^2}{48} - C (\ln B)^{3/2} \,,
$$
which is the claimed lower bound in Theorem \ref{thm:operator}.
\end{proof}


\section{Proof of Proposition \ref{pro:fockmain}}
\label{sec:fockmain}

In this section we establish Proposition \ref{pro:fockmain} which lies at the heart of Theorem \ref{thm:operator}. The proof consists of two main steps. In the first step we replace the phonon field $a_k$, $k\in\R^3$, by an effective phonon field $\hat a_{k_3}$, which only depends on a \emph{one-dimensional} parameter $k_3$. Moreover, the electron-phonon coupling is changed from $|k|^{-1}$ to an effective coupling $v(k_3)$, which is almost constant and grows logarithmically with $B$. The precise statement is given in Lemma \ref{lem:quant1d}. The second step in the proof of Proposition \ref{pro:fockmain} is the analysis of an essentially one-dimensional problem, see Subsections \ref{ssec:locdec} and \ref{ssec:errors}. Here we can follow a one-dimensional version of the strategy developed in \cite{LieTho97}; see also \cite{Gha12}.

We recall the estimate on $\mathfrak{h}_\K^\co$  from \eqref{eq:simplelowerop}. We begin our lower bound on $P_0\mathfrak{h}_\K^\co P_0$ by introducing ultra-violet cut-offs, similarly as in the proof of Lemma \ref{lem:firstcutoff}. The fact that $\mathfrak{h}_\K^\co$ is sandwiched by projections $P_0$, however, allows us to choose these cut-offs more carefully and, in particular, to distinguish between the directions $k_\bot$ and $k_3$.

\subsection{Projection onto $P_0$}
\label{ssec:proj}

We begin by deriving a convenient representation of $P_0 e^{ik_\bot \cdot x_\bot} P_0$. For $k_\perp \in \R^2$ let us define the integral operator $I_{k_\perp}$ in $L^2(\R^2)$ with integral kernel
$$
I_{k_\perp} (x_\perp,y_\perp) = P_0(x_\perp,y_\perp) e^{k_\perp \wedge (x_\perp-y_\perp)/2} e^{i k_\perp \cdot (x_\perp+y_\perp)/2 } \, .
$$
Again we also write $I_{k_\perp}$ for the operator $I_{k_\perp} \otimes 1 \otimes 1$ on $L^2(\R^2)\otimes L^2(\R) \otimes \mathcal{F}$.

\begin{lemma}
\label{lem:k3proj}
For $k_\perp \in \R^2$ we have
$$
P_0 e^{ik_\perp \cdot x_\perp} P_0 = P_0 e^{-|k_\perp|^2/2B} I_{k_\perp} P_0 \, .
$$
Moreover, the operator $I_{k_\perp}$ is bounded with $\| I_{k_\perp} \| \leq 2 e^{|k_\perp|^2/4B}$.  
\end{lemma}

\begin{proof}
The first claim follows from \eqref{eq:prokernel}, which leads to the identity
$$
\int_{\R^2} P_0(z_\perp,x_\perp)  e^{ik_\perp \cdot x_\perp} P_0(x_\perp,y_\perp) dx_\perp = e^{-|k_\perp|^2/2B} I_{k_\perp} (z_\perp,y_\perp) \, .
$$
To prove the second claim we estimate $| I_{k_\perp}(x_\perp,y_\perp)|  \leq  \tilde I_{k_\perp} (x_\perp-y_\perp)$, where
$$
\tilde I_{k_\perp}(x_\perp) = \frac{B}{2\pi} e^{-B|x_\perp|^2/4} e^{k_\perp \wedge x_\perp/2} \, .
$$
Hence, by Young's inequality, for $\phi\in L^2(\R^2)$,
$$
\| I_{k_\perp} \phi \|^2 \leq \| \tilde I_{k_\perp} \ast |\phi| \|^2 \leq \| \tilde I_{k_\perp} \|_{L^1(\R^2)}^2 \| \phi \|^2 \,,
$$
and the claim follows from the estimate $\| \tilde I_{k_\perp} \|_{L^1(\R^2)} \leq 2 e^{|k_\perp|^2/4B}$.
\end{proof}

In view of Lemma \ref{lem:k3proj} we can write, for all $k \in \R^3$,
\begin{equation}
\label{eq:projres}
 P_0  \frac{a_k}{|k|}  e^{i k \cdot x}  P_0 = P_0  \frac{a_k}{|k|} I_{k_\perp} e^{-|k_\perp|^2/2B} e^{i k_3 x_3}  P_0 
\end{equation}
and similarly for the hermitian conjugate.

\subsection{Ultraviolet cutoff}
\label{ssec:cutoff}

Now we first cut off high phonon modes in the $k_3$-coordinate and then we cut off high and low modes in the first two coordinates. The results from Subsection~\ref{ssec:proj} allow us to choose these cut-offs in a more precise way than in Section \ref{sec:fock}. In particular, we can restrict to phonon modes $k \in \R^3$ with $|k_\perp| \leq \K_\perp$ and $|k_3| \leq \K_3$ with positive parameters $\K_\perp$ and $\K_3$, both smaller than $\K$. This explains  the assumption $\K \geq \sqrt B$ in the proposition. Eventually, we will choose $\K_\perp$ of order $\sqrt B$ and $\K_3$ to be comparable to a power of $\ln B$ (recall that we chose $\K = B (\ln B)^{-4/3}$  in the proof of Theorem \ref{thm:operator}). 

\begin{lemma}
\label{lem:ultra}
For $0 < \K_3 \leq \K$ we have
\begin{align*}
P_0 \mathfrak{h}_{\K}^\co P_0 \geq \, & \kappa P_0 H_B P_0 +  \kappa_1 P_0 (-\partial_3^2) P_0\\
&  + P_0 \int_{|k_3| \leq \K_3, |k_\perp|\leq \K}  \lk a_k^* a_k +   \frac{\sqrt \alpha}{2\pi} \frac{a_k}{|k|}  e^{ik \cdot x} + \frac{\sqrt \alpha}{2\pi} \frac{a_k^*}{|k|}  e^{-ik \cdot x}  \rk dk  P_0 - \frac 12 \, ,
\end{align*}
where
$$
\kappa_1 =  \kappa - \frac{8 \alpha}{\pi \K_3} \int_0^\infty e^{-\K^2_3 t /2B } \frac 1 {1+t} dt \,.
$$
\end{lemma}

\begin{proof}
Similar as in the proof of Lemma \ref{lem:roughlower} we set 
$$
Z = \frac{\sqrt \alpha}{2\pi} \int \frac{a_k}{|k| k_3} I_{k_\perp} e^{-|k_\perp|^2/2B} e^{ik_3 x_3} dk \,.
$$
Here and in the remainder of the proof all integrals are over $\{k \in \R^3 : |k_\perp| \leq \K, \K_3 \leq |k_3| \leq \K \}$, unless stated otherwise. We can write, in view of \eqref{eq:projres},
\begin{align}
\label{eq:projres1}
& P_0 \frac{\sqrt \alpha}{2\pi}  \int \lk \frac{a_k}{|k|} e^{ik\cdot x} + \frac{a_k^*}{|k|} e^{-ik\cdot x} \rk dk P_0 = \\
& \qquad = P_0 \frac{\sqrt \alpha}{2\pi}  \int \lk \frac{a_k}{|k|} I_{k_\perp} e^{-|k_\perp|^2/2B} e^{ik_3 x_3} + \frac{a_k^*}{|k|} I_{k_\perp}^* e^{-|k_\perp|^2/2B} e^{-ik_3 x_3} \rk dk P_0 \notag \\
& \qquad = P_0 [-i\partial_3,Z-Z^*] P_0 \,. \notag
\end{align}
In the same way as in the proof of Lemma \ref{lem:firstcutoff} we obtain the estimate 
\begin{equation}
\label{eq:ultra2}
-\frac{\sqrt \alpha}{2\pi}  \int \lk \frac{a_k}{|k|} I_{k_\perp} e^{-|k_\perp|^2/2B} e^{ik_3 x_3} + \frac{a_k^*}{|k|} I_{k_\perp}^* e^{-|k_\perp|^2/2B} e^{-ik_3 x_3} \rk dk \leq  \rho ( -\partial_3^2) + \frac 2\rho \lk Z^* Z + Z Z^* \rk
\end{equation}
for any $\rho > 0$. We claim that
\begin{equation}
\label{eq:ultraz}
 Z^* Z + Z Z^*  \leq \frac{\tilde R_1}{2}    \lk \int  a_k^* a_k   dk+ \frac 1 2 \rk  \,,
 \qquad
 \tilde R_1 = \frac{8 \alpha}{\pi \K_3} \int_0^\infty e^{-\K^2_3 t /2B } \frac 1 {1+t} dt
 \, .
\end{equation}

Indeed, the expectation $\langle\cdot\rangle$ in any normalized state satisfies
$$
\langle Z^*Z \rangle = \frac{\alpha}{4 \pi^2} \iint \frac{1}{k_3 k_3' |k| |k'|} e^{-(|k_\perp|^2+|k_\perp'|^2)/2B} e^{i(k_3'-k_3)x_3} \langle a_k^* I_{k_\perp}^* I_{k'_\perp} a_{k'} \rangle dk dk' \, .
$$
Combining the Schwarz inequality and Lemma \ref{lem:k3proj} yields
$$
\langle a_k^* I_{k_\perp}^* I_{k'_\perp} a_{k'} \rangle \leq \langle a_k^* I_{k_\perp}^* I_{k_\perp} a_k \rangle^{1/2} \langle a_{k'}^* I_{k'_\perp}^* I_{k'_\perp} a_{k'} \rangle^{1/2} \leq 4 e^{(|k_\perp|^2+|k_\perp'|^2)/4B}\langle a_k^*  a_k \rangle^{1/2} \langle a_{k'}^* a_{k'} \rangle^{1/2} \, .
$$
It follows that
$$
\langle Z^*Z \rangle \leq \frac{\alpha}{\pi^2} \lk  \int \frac{1}{|k_3| |k|} e^{-|k_\perp|^2/4B} \langle a_k^*  a_k \rangle^{1/2} dk \rk^2 \leq \frac{\alpha}{\pi^2} \int \frac{1}{k_3^2 |k|^2} e^{-|k_\perp|^2/2B} dk \int  \langle a_k^*  a_k \rangle dk
$$
with
\begin{align*}
\int \frac{1}{k_3^2 |k|^2} e^{-|k_\perp|^2/2B} dk 
& \leq \frac{4\pi}{\K_3} \int_0^\infty e^{-\K_3^2 r^2 /2B } \frac 1 {r^2} \lk r + \arctan \lk \frac 1r \rk - \frac \pi 2 \rk dr \\
&\leq  \frac{4\pi}{\K_3} \int_0^\infty e^{-\K_3^2 r^2 /2B } \frac r {1+r^2}  dr \\
& = \frac{2\pi}{\K_3} \int_0^\infty e^{-\K_3^2 t /2B } \frac 1 {1+t}  dt 
= \frac{\pi^2}{4\alpha} \tilde R_1 \, .
\end{align*}
To estimate $\langle Z Z^* \rangle$ we note that $\langle a_k I_{k_\perp} I_{k'_\perp}^* a_{k'}^* \rangle = \langle a_{k'}^* I_{k_\perp} I^*_{k'_\perp} a_k \rangle + \langle I_{k_\perp} I^*_{k_\perp} \rangle \delta(k-k')$ and we can argue in the same way as above. This establishes \eqref{eq:ultraz}.

Now we choose $\rho=\tilde R_1$, such that by \eqref{eq:ultraz}
$$
\frac 2\rho \lk Z^* Z + Z Z^* \rk \leq  \int  a_k^* a_k  dk + \frac 12  \, .
$$
The bound claimed in the lemma then follows from \eqref{eq:simplelowerop} and \eqref{eq:ultra2} after projecting onto the range of $P_0$ and recalling \eqref{eq:projres1}.
\end{proof}

Next, we cut off high and low phonon modes in the first two coordinates.

\begin{lemma}
\label{lem:cutoff}
For any $0 < \K_3 \leq \K$ and $1 \leq \K_\perp \leq \K$ we have
$$
P_0 \mathfrak{h}_\K^{\co} P_0 \geq  P_0 \left[ \kappa_2 H_B +  \kappa_1 (-\partial_3^2) + \frac{\sqrt \alpha}{2\pi} \int_\Omega  \lk  a_k^* a_k + \frac{a_k}{|k|}  e^{ik \cdot x} + \frac{a_k^*}{|k|}  e^{-ik \cdot x}  \rk dk - \lk 1+ \frac \alpha 2 \rk \right] P_0  \, ,
$$
with $\Omega = \{ k\in\R^3:\ |k_3| \leq \K_3, 1\leq |k_\perp| \leq \K_\perp  \}$ and with $\kappa_1$ from Lemma \ref{lem:ultra} and 
$$
\kappa_2 =  \kappa - 2\alpha \pi^{-1} \K_3 \K_\perp^{-2} \,.
$$
\end{lemma}

For this result it is important that we have already cut off high modes in the $k_3$-direction. Otherwise, the bound on $\kappa-\kappa_2$ would be $C/\K_\perp$, similarly as in Lemma \ref{lem:firstcutoff}. This makes a difference since eventually we want to choose $\K_\perp$ of order $\sqrt B$ and $\K_3$ to be comparable to a power of $\ln B$.

\begin{proof}
We continue with the lower bound given in Lemma \ref{lem:ultra}. We need to bound the contribution of the modes from $\Omega' = \{ k\in\R^3:\ |k_3| \leq \K_3, \K_\perp < |k_\perp| \leq \K \}$ and $\Omega''= \{k\in\R^3:\ |k_3| \leq \K_3 , |k_\perp| < 1\}$.

We begin with $\Omega'$. Similar as in Lemma~\ref{lem:firstcutoff} we set, for $j = 1,2$, 
$$
Z_j = \frac{\sqrt \alpha}{2\pi} \int_{\Omega' } \frac{k_j}{|k|^3} a_k e^{ik\cdot x} dk \, .
$$
Then for any $\epsilon > 0$
\begin{equation}
\label{eq:cutoff1}
- \frac{\sqrt \alpha}{2\pi}  \int_{\Omega'} \lk\frac{a_k}{|k|} e^{ik\cdot x}+\frac{a_k^*}{|k|} e^{-ik\cdot x} \rk dk \leq \epsilon \langle H_B \rangle + \frac 2{\epsilon} \langle \Zv^* \cdot \Zv + \Zv \cdot \Zv^* \rangle \, ,
\end{equation}
where $\Zv$ denotes the vector $(Z_1,Z_2)$.
Similar as above we can estimate
$$
 \Zv^* \cdot \Zv + \Zv \cdot \Zv^*  \leq \frac{\alpha}{2 \pi^2} \lk \int_{\Omega'}  a_k^* a_k  dk + \frac 12 \rk \int_{\Omega'} \frac{|k_\perp|^2}{|k|^6} dk \, .
$$
and we note that 
\begin{align*}
\int_{\Omega'} \frac{|k_\perp|^2}{|k|^6} dk 
&\leq 2 \int_{|k_\perp| > \K_\perp} \frac{1}{|k_\perp|^3} \int_0^{\K_3/|k_\perp|} \frac{1}{(1+t^2)^3} dt dk_\perp  \\
& \leq 2\int_{|k_\perp| > \K_\perp} \frac{\K_3}{|k_\perp|^4} dk_\perp\\
& = 2\pi \K_3 \K_\perp^{-2} \, .
\end{align*}
Hence, we choose $\epsilon= 2\alpha \pi^{-1} \K_3 \K_\perp^{-2}$ such that 
$$
\frac 2{\epsilon}  \lk \Zv^* \cdot \Zv + \Zv \cdot \Zv^* \rk \leq   \int_{\Omega'}  a_k^* a_k  dk + \frac 12  \, .
$$

We combine this estimate with \eqref{eq:cutoff1} to get
$$
\int_{\Omega'} \lk  a_k^* a_k + \frac{\sqrt \alpha}{2\pi} \frac{a_k}{|k|} e^{ik\cdot x} + \frac{\sqrt \alpha}{2\pi} \frac{a_k^*}{|k|} e^{-ik\cdot x} \rk dk \geq -  2\alpha \pi^{-1} \K_3 \K_\perp^{-2}  H_B  - \frac 12 \, .
$$

To complete the proof it remains to estimate the contribution of low modes in $\Omega''$. For all $k \in \R^3$ we have $\lk a_k^*+ \sqrt \alpha e^{i k \cdot x}/2 \pi |k| \rk \lk a_k + \sqrt \alpha e^{-i k \cdot x}/ 2 \pi |k| \rk \geq 0$ and therefore
\begin{align*}
\int_{\Omega''} \lk  a_k^* a_k + \frac{\sqrt \alpha}{2\pi} \frac{a_k}{|k|} e^{i k \cdot x} + \frac{\sqrt \alpha}{2\pi} \frac{a_k^*}{|k|} e^{-i k \cdot x} \rk dk
& \geq - \frac{\alpha}{4\pi^2} \int_{\Omega''} \frac 1{|k|^2} \,dk \\
& \geq - \frac{\alpha}{4 \pi^2}  \int_{|k_\perp|<1} \frac{1}{|k_\perp|} dk_\perp \int_\R \frac 1{1+t^2} dt \\
& = - \frac \alpha 2\, .
\end{align*}
This finishes the proof.
\end{proof}

\subsection{Reduction to one dimension}
\label{ssec:1d}

We are now ready to state our first main result in this section, namely a lower bound on $P_0\mathfrak{h}_\K^{\co} P_0$ in terms of an essentially one-dimensional operator.

\begin{lemma}
\label{lem:quant1d}
For $0 < \K_3 \leq \K$ and $1 \leq \K_\perp \leq \K$ let 
$$
v(k_3) = \lk \int_{1 \leq |k_\perp| \leq \K_\perp} |k|^{-2} dk_\perp \rk^{1/2} = \sqrt \pi \lk \ln \lk \K_\perp^2 + k_3^2 \rk - \ln \lk 1 + k_3^2 \rk \rk^{1/2} \, .
$$
There are creation and annihilation operators $\al^*$ and $\al$ on $\mathcal{F}(L^2(\R^3))$ with
$$
[ \al, \all^*] = \delta(k_3- k'_3)\,,
\quad [ \al, \all]= [ \al^*, \all^*] = 0
\qquad\text{for}\ k_3, k_3'\in\R
$$
such that the operator
$$
\mathfrak{h}^{\textnormal{1d}} = \kappa_1 (- \partial_3^2) + \int_{|k_3| \leq \K_3}  \al^* \al dk + \frac{\sqrt \alpha}{2\pi} \int_{|k_3| \leq \K_3} v(k_3) \lk  \al  e^{i k_3 x_3} + \al^* e^{-i k_3 x_3} \rk dk_3
$$
satisfies for $B>0$ the estimate  
$$
P_0 \mathfrak{h}_{\K}^{\co} P_0 \geq \kappa_2 B  P_0 + P_0 \mathfrak{h}^{\textnormal{1d}} P_0 - C P_0
$$
with $\kappa_1$ and $\kappa_2$ from Lemmas \ref{lem:ultra} and \ref{lem:cutoff} and $C=1+  \alpha/2$.
\end{lemma}

\begin{proof}
We introduce
$$
\al = \frac{1}{v(k_3)} \int_{1 \leq |k_\perp| \leq \K_\perp} \frac{a_k}{|k|} e^{i k_\perp \cdot x_\perp} dk_\perp \,.
$$
To verify the commutation relations for $\al$ and $\al^*$ we write
$$
\left[ \al, \al^* \right] = \frac{1}{v(k_3) v(k_3')} \int_{1 \leq |k_\perp| \leq \K_\perp} \int_{1 \leq |k'_\perp| \leq \K_\perp} \frac 1{|k| |k'|} \left[a_k, a_{k'}^* \right] e^{i(k_\perp-k'_\perp) \cdot x_\perp} dk_\perp' dk_\perp
$$
and note that the right-hand side equals $\delta(k_3-k_3')$, by definition of $v(k_3)$ and the fact that $[a_k, a^*_{k'} ] = \delta(k-k')$. The other relations are verified similarly.

Next, by means of the Schwarz inequality, applied similarly as in the proof of Lemma \ref{lem:ultra}, we learn that
$$
\al^* \al  \leq \int_{1 \leq |k_\perp| \leq \K_\perp} a_k^* a_k \,dk_\perp \,.
$$
The assertion now follows from Lemma \ref{lem:cutoff} together with the fact that $P_0 H_B P_0 = B P_0$.
\end{proof}

\subsection{Localization and decomposition}
\label{ssec:locdec}

With Lemma \ref{lem:quant1d} at hand we can essentially follow the strategy of \cite{LieTho97} to complete the proof of Proposition \ref{pro:fockmain}.

First, we localize the electron in the $x_3$-direction in intervals of length $L>0$, where $L$ is a parameter that will be specified later. We fix a non-negative function $\chi \in C_0^\infty(\R)$ with support in the interval $[-1/2,1/2]$ that satisfies $\int_\R \chi^2(t) dt = 1$.  For $u \in \R$ we put 
$$
\chi_u(t) =  \frac{1}{\sqrt L} \chi \lk \frac{t-u}{L} \rk \, . 
$$
Then $\chi_u$ is supported in the interval $[u - L/2, u + L/2]$ and  satisfies $\int_\R \chi_u^2(t) dt =1$. Moreover, for all fixed $t \in \R$ we have $\int_\R | \chi_u'(t)|^2 du = \| \chi' \|^2 L^{-2}$.

We note that for all $f,g \in C_0^\infty(\R)$
$$
\frac 12 \lk f,  g^2 (-f'') \rk_{L^2(\R)} + \frac 12 \lk f, (- g^2 f)'' \rk_{L^2(\R)} = \lk g f, (-g f)'' \rk_{L^2(\R)} - \lk f, f (g')^2 \rk_{L^2(\R)} \, .
$$
Applying this identity with $g = \chi_u$ implies
$$
- \partial_3^2 = \int_\R  \chi_u (-\partial_3^2) \chi_u \, du - \|\chi'\|^2 L^{-2}
$$
and
\begin{equation}
\label{eq:locres}
\mathfrak{h}^{\textnormal{1d}} =
\int_\R  \chi_u \mathfrak{h}^{\textnormal{1d}} \chi_u \, du - \|\chi'\|^2 L^{-2} \, .
\end{equation}

After localizing the electron we decompose the phonon modes in the $k_3$-coordinate into $M$ intervals of length $P=2\K_3/M$, where $M\in\N$ is a parameter to be chosen later. We label these intervals by $b$. We want to group together modes $\al$ that belong to an interval $b$. To do this we have to replace the factors $e^{ik_3 x_3}$ by factors independent of $k_3$. So for each $b$, we choose a value $k_b \in \R$ in the block $b$. (Later on, we will optimize over $k_b$, but the bound in the following lemma is true uniformly for any choice.) Then we get the following estimate.

\begin{lemma}
\label{lem:block}
For every $u\in\R$ there are creation and annihilation operators $A_b^{(u)*}$ and $A_b^{(u)}$ on $\mathcal{F}(L^2(\R^3))$ satisfying
$$
\left[ A_b^{(u)}, A_{b'}^{(u)*}\right] = \delta_{bb'}\,,
\quad \left[ A_b^{(u)}, A_{b'}^{(u)}\right]= \left[ A_b^{(u)*}, A_{b'}^{(u)*}\right] = 0
\qquad\text{for all blocks}\ b, b'
$$
with the following property. For any $u\in\R$ and $0 < \gamma < 1$ the operator
$$
\mathfrak{h}_{\gamma}^{(u)} = \kappa_1 (- \partial_3^2) + \sum_b \left[ (1-\gamma) A_b^{(u)*} A_b^{(u)} + \frac{\sqrt \alpha}{2\pi} V(b) \lk A_b^{(u)}   e^{i k_b x_3} +  A_b^{(u)*} e^{-i k_b x_3} \rk \right]
$$
with $V(b) = (\int_b v(k_3)^2  dk_3)^{1/2}$ satisfies
$$
\chi_u \mathfrak{h}^{\textnormal{1d}} \chi_u \geq   \chi_u \mathfrak{h}_{\gamma}^{(u)} \chi_u -   C \frac{\K^2_3 L^2}{\gamma M^2} R \chi_u^2
$$
with an error term
$$
R = \int_{|k_3|\leq \K_3} v(k_3)^2 dk_3 =  \pi  \int_{|k_3| \leq \K_3} \lk \ln \lk \K_\perp + k_3^2 \rk - \ln \lk 1 + k_3^2 \rk \rk  dk_3 \, .
$$
\end{lemma}

\begin{proof}
For all $k_3 \in b$ and all $x_3$ in the support of $\chi_u$ we have
\begin{equation}
\label{eq:expest}
\left| e^{ik_3 (x_3-u)} - e^{i k_b (x_3-u)} \right| \leq |(k_3 - k_b) (x_3-u)| \leq \frac {P L}2 \, . 
\end{equation}
To replace $k_3$ by $k_b$ in the definition of $\mathfrak h^{\textnormal{1d}}$ we introduce a small parameter $\gamma >0$ and estimate 
\begin{align*}
& \frac{\sqrt \alpha}{2\pi} \sum_b \!\int_b \! v(k_3)\! \lk\! \lk e^{ik_3 (x_3-u)} - e^{i k_b (x_3-u)} \rk  e^{ik_3u} \al \! + \! \lk e^{-ik_3 (x_3-u)} - e^{-i k_b (x_3-u)} \rk e^{-ik_3u} \al^* \rk \! dk_3 \\
& + \gamma \sum_b \int_b \al^* \al dk_3
\end{align*}
from below. 
We complete the square and using \eqref{eq:expest} we obtain a lower bound
$$
 - \frac{P^2 L^2 \alpha}{16 \pi^2 \gamma}  \sum_b \int_b v(k_3)^2 dk_3 =  - \frac{\K^2_3 L^2 \alpha}{4 \pi^2 \gamma M^2} R \, .
$$
This allows us to estimate
\begin{align}
\nonumber
\chi_u \mathfrak{h}^{\textnormal{1d}} \chi_u \geq & \kappa_1 \chi_u (-\partial_3^2) \chi_u + (1-\gamma) \sum_b \int_b \al^* \al \,dk_3 \chi_u^2 \\
\nonumber
& + \frac{\sqrt \alpha}{2\pi} \sum_b \int_b v(k_3) \lk e^{i(k_3-k_b)u} \al e^{ik_b x_3} + e^{-i(k_3-k_b)u} \al^* e^{-ik_b x_3} \rk dk_3 \chi_u^2 \\
\label{eq:blockres}
& - C  \frac{\K^2_3 L^2}{\gamma M^2} R \chi_u^2  \, .
\end{align}

For each block $b$ we define new creation and annihilation operators $A_b^{(u)*}$ and $A_b^{(u)}$ by
$$
A_b^{(u)} = \frac 1{V(b)} \int_b v(k_3) e^{i(k_3-k_b)u} \al dk_3 \, .
$$
Similar as in the proof of Lemma \ref{lem:quant1d} we can show that these operators satisfy the commutation relations and $A_b^{(u)*} A_b^{(u)} \leq \int_b \al^* \al \,dk_3$. Inserting this into \eqref{eq:blockres} completes the proof.
\end{proof}

\subsection{Error estimates}
\label{ssec:errors}

In view of Lemma \ref{lem:block} it suffices to analyze the operator $\mathfrak{h}_\gamma^{(u)}$. We emphasize that the operators $A_b^{(u)}$ and $A_b^{(u)*}$ introduced above are properly normalized boson modes. Thus we can use coherent states as in \cite{LieTho97} to estimate the operator $\mathfrak{h}_\gamma^{(u)}$.

We work under the assumptions 
\begin{equation}
\label{generalassumptions}
0 < \K_3 \leq \K \, , \quad 1 \leq \K_\perp \leq \K \, , \quad \textnormal{and} \quad \kappa_1 > 0
\end{equation}
and emphasize that these will be satisfied by our final choice of parameters.

Similarly as in \cite{LieTho97} (see also \cite{Gha12} for the one-dimensional case) we obtain (for a suitable choice of $k_b$)
\begin{equation}
\label{eq:coh}
\mathfrak{h}_\gamma^{(u)}  \geq I - M  \, ,
\end{equation}
where
$$
I = \inf_{\|\phi\|=1}  \left[\kappa_1 \| \partial_3 \phi\|^2 
- \frac \alpha {4\pi^2(1-\gamma)} \int_\R v(k_3)^2 \left| \int e^{ik_3 x_3} |\phi(x_\perp,x_3)|^2 dx \right|^2 dk_3 \right] \,.
$$
Combining \eqref{eq:coh} with the localization formula \eqref{eq:locres} and with Lemma~\ref{lem:block} we obtain
$$
\mathfrak{h}^{\textnormal{1d}} \geq \lk I-M -C \frac{\K^2_3 L^2}{\gamma M^2} R \rk \int_\R \chi_u^2 \,du - \|\chi'\|^2 L^{-2} = I-M - C\frac{\K^2_3 L^2}{\gamma M^2} R - \|\chi'\|^2 L^{-2} \,.
$$
Finally, combining this with Lemma \ref{lem:quant1d} and the expression for $\kappa_2$ we see that 
\begin{align}
\label{eq:cohres}
 P_0 \mathfrak{h}_\K^{\co} P_0  \geq & \, \kappa B P_0 + I P_0 - C \lk \frac{\K_3 B}{\K_\perp^2} + M +   \frac{\K^2_3 L^2}{\gamma M^2} R + \frac{1}{L^2} \rk P_0  \, .
\end{align}

Our next step is to estimate $I$ from below. To do so we insert the bound
\begin{equation}
\label{eq:vbound}
v(k_3)^2 = \pi \lk \ln \lk \K_\perp^2 + k_3^2 \rk - \ln \lk 1 + k_3^2 \rk \rk \leq 2 \pi \ln \K_\perp 
\end{equation}
into the infimum defining $I$ and perform the $k_3$-integration to get
\begin{align*}
I\geq \inf_{\|\phi\|=1}  \left[\kappa_1 \| \partial_3 \phi\|^2 
- \frac{ \alpha \ln \K_\perp}{1-\gamma} \int_\R \left| \int_{\R^2} |\phi(x_\perp,x_3)|^2 dx_\perp \right|^2 dx_3 \right] = -\frac{\alpha^2 (\ln\K_\perp)^2}{12\kappa_1 (1-\gamma)^2} \,. 
\end{align*}
The last identity used Corollary \ref{cor:secondterm}. Now under the assumptions
\begin{equation}
\label{eq:choosepar}
\kappa-\kappa_1 \leq \frac{\kappa}{2}
\quad\text{and}\quad
\gamma\leq \frac12
\end{equation}
there is a constant $C>0$ such that we can estimate the right side further by
\begin{align}
\label{eq:inf}
I \geq -\frac{\alpha^2 (\ln\K_\perp)^2}{12\kappa} - C \frac{(\ln\K_\perp)^2}{\kappa} \lk \frac{\kappa-\kappa_1}{\kappa} + \gamma \rk \,.
\end{align}

Thus, summarizing \eqref{eq:cohres} and \eqref{eq:inf} we have, assuming \eqref{generalassumptions} and \eqref{eq:choosepar},
\begin{align*}
 P_0 \mathfrak{h}_\K^{\co} P_0  \geq & \lk  \kappa B -\frac{\alpha^2 (\ln\K_\perp)^2}{12\kappa} \rk P_0 \\
 &- C \lk \frac{(\ln\K_\perp)^2}{\kappa} \lk \frac{\kappa-\kappa_1}{\kappa} + \gamma \rk + \frac{\K_3 B}{\K_\perp^2} + M +   \frac{\K^2_3 L^2}{\gamma M^2} R + \frac{1}{L^2} \rk P_0 \, .
\end{align*}
We use \eqref{eq:vbound} again to bound
$$
R  = \int_{|k_3| \leq \K_3} v(k_3)^2 dk_3 \leq 2\pi \K_3 \ln \K_\perp 
$$ 
and optimize the resulting expression with respect to $L$, $M$ and $\gamma$ by choosing
$L^{2} = \kappa^{1/5} \K_3^{-3/5}(\ln\K_\perp)^{-3/5}$, $M=[L^{-2}]$ and $\gamma =\kappa^{4/5} \K_3^{3/5} (\ln\K_\perp)^{-7/5}$ to get
$$
 P_0 \mathfrak{h}_\K^\co P_0  \geq  \lk \kappa B -\frac{\alpha^2 (\ln\K_\perp)^2}{12\kappa} \rk P_0 - C \lk \frac{(\kappa-\kappa_1)(\ln\K_\perp)^2}{\kappa} + \frac{\K_3^{3/5} (\ln\K_\perp)^{3/5}}{\kappa^{1/5}} + \frac{\K_3 B}{\K_\perp^2} \rk P_0
$$
provided $\kappa^{4/5} \K_3^{3/5} (\ln\K_\perp)^{-7/5}\leq1/2$ and $\kappa - \kappa_1 \leq \kappa /2$.

A simple bound shows that, as long as $B/ \K_3^2\geq 2$,
$$
\kappa-\kappa_1 = \frac{8 \alpha}{\pi \K_3} \int_0^\infty e^{-\K^2_3 t /2B } \frac 1 {1+t} dt \leq \frac{C}{\K_3} \ln \frac{B}{\K_3^2} \,.
$$ 

We optimize the remaining three error terms with respect to $\K_3$ and $\K_\perp$ (under the additional assumption that $\kappa$ is close to one). In particular, we choose $\K_\perp = B^{1/2}$ and $\K_3 = \kappa^{-1/2} (\ln B)^{3/2}$ and verify that the conditions \eqref{generalassumptions}, \eqref{eq:choosepar}, and $B/\K_3^2\geq 2$ are satisfied for $B$ large enough. (At this point we use the assumptions $C(\ln B)^{-1/2}\leq\kappa\leq C^{-1}\ln B$ and $\K \geq \sqrt B$.) We obtain
$$
P_0 \mathfrak{h}_\K^{\co} P_0  \geq \lk \kappa B -\frac{\alpha^2 (\ln B)^2}{48 \kappa} - C \kappa^{-1/2} (\ln B)^{3/2} \rk P_0 \,,
$$
which is the bound claimed in Proposition \ref{pro:fockmain}. The proof is complete.



\begin{thebibliography}{GHW12}

\bibitem[AHS81]{AvHeSi}
J. E. Avron, I. W. Herbst, B. Simon, \textit{Schr\"odinger operators with magnetic fields III. Atoms in homogeneous magnetic field}. Comm. Math. Phys. \textbf{79} (1981), 529--572. 

\bibitem[BSY00]{BaSoYn}
B. Baumgartner, J. P. Solovej, J. Yngvason, \textit{Atoms in strong magnetic fields: The high field limit at fixed nuclear charge}. Commun. Math. Phys. \textbf{212} (2000), 703--724.

\bibitem[Dev96]{De}
J.T. Devreese, \textit{Polarons}. In: Encyclopedia of Applied Physics \textbf{14}, edited by G. L. Trigg (VCH Publishers, Weinheim, 1996), 383--413.

\bibitem[GL91]{GeLo}
B. Gerlach, H. L\"owen, \textit{Analytical properties of polaron systems or: do polaronic phase transitions exist or not?}. Rev. Mod. Phys. \textbf{63} (1991), no. 1, 63--90.

\bibitem[Gha12]{Gha12}
R.~Ghanta, \emph{Exact ground state energy of the 1{D} strong-coupling polaron}, Junior thesis, Princeton, 2012.

\bibitem[GHW12]{GrHaWe}
M. Griesemer, F. Hantsch, D. Wellig, \textit{On the magnetic Pekar functional and the existence of bipolarons}. Rev. Math. Phys. \textbf{24} (2012), no. 6, 1250014.

\bibitem[Gro76]{Gr}
E. P. Gross, \textit{Strong coupling polaron theory and translational invariance}. Ann. of Phys. \textbf{99} (1976), 1--29.

\bibitem[KLS92]{KoLeSm}
E. A. Kochetov, H. Leschke, M. A. Smondyrev, \textit{Diagrammatic weak-coupling expansion for the magneto-polaron energy}. Z. Phys. B - Condensed Matter \textbf{89} (1992), 177--186.

\bibitem[Kuk73]{Ku}
L. S. Kukushkin, Fiz. Tverd. Tela 15, 859 (1973); Sov. Phys. - Solid State 15, 591 (1973)

\bibitem[LL76]{LanLif76}
L.~D. Landau, E.~M. Lifshitz, \emph{Course of theoretical physics. {V}ol.
  1}, third ed., Pergamon Press, Oxford, 1976.

\bibitem[LM76]{LeMa}
Y. Lepine, D. Matz, \textit{Fock approximation to the large polaron in a magnetic field}. Can. J. Phys. \textbf{54} (1976), 1979--1989.

\bibitem[Lie76]{Li}
E. H. Lieb, \textit{Existence and uniqueness of the minimizing solution of Choquard's nonlinear equation}. Studies in Appl. Math. \textbf{57} (1976/77), no. 2, 93--105.

\bibitem[LL01]{LieLos01}
E.~H. Lieb and M. Loss, \emph{Analysis}, second ed., Graduate Studies
  in Mathematics, vol.~14, American Mathematical Society, Providence, RI, 2001.

\bibitem[LSY94]{LieSolYng94}
E.~H. Lieb, J.~P. Solovej, J. Yngvason, \emph{Asymptotics of
  heavy atoms in high magnetic fields. {I}. {L}owest {L}andau band regions},
  Comm. Pure Appl. Math. \textbf{47} (1994), no.~4, 513--591.

\bibitem[LT97]{LieTho97}
E.~H. Lieb, L.~E. Thomas, \emph{Exact ground state energy of the
  strong-coupling polaron}, Comm. Math. Phys. \textbf{183} (1997), no.~3,
  511--519.

\bibitem[LY58]{LieYam58}
E.~H. Lieb, K.~Yamazaki, \emph{Ground-state energy and effective mass of the
  polaron}, Phys. Rev. \textbf{111} (1958), 728--733.

\bibitem[MS07]{MiSp} 
T. Miyao, H. Spohn, {\it The bipolaron in the strong coupling limit}. Ann. Henri Poincar{\'e} {\bf 8} (2007), 1333--1370.

\bibitem[Nel64]{Ne}
E. Nelson, \textit{Interaction of non-relativistic particles with a quantized scalar field}. J. Math. Phys. \textbf{5} (1964), 1190--1197.

\bibitem[Pek63]{Pek63}
S.~I. Pekar, \emph{Research in electron theory of crystals}, United States
  Atomic Energy Commission, Washington, DC, 1963.

\bibitem[PT51]{Pek51}
S.~I. Pekar, O.~F. Tomasevich, \emph{Theory of {F} centers}, Zh. Eksp. Teor.
  Fys. \textbf{21} (1951), 1218--1222.

\bibitem[Sai81]{Sa}
M. Saitoh, \textit{Free energy of a polaron in a strong magnetic field}. J. Phys. Soc. Jpn. \textbf{50} (1981), no. 7, 2295--2302.

\bibitem[vSN41]{Nag41}
Bela v.~Sz.~Nagy, \emph{\"{U}ber {I}ntegralungleichungen zwischen einer
  {F}unktion und ihrer {A}bleitung}, Acta Univ. Szeged. Sect. Sci. Math.
  \textbf{10} (1941), 64--74.

\bibitem[WPR76]{WhPaRo}
G. Whitfield, R. Parker, M. Rona, \textit{Adiabatic approximation for a polaron in a magnetic field}. Phys. Rev. B \textbf{13} (1976), no. 5, 2132--2137.

\end{thebibliography}
\end{document}